\documentclass[11pt]{scrartcl}

\usepackage{url}
\usepackage[hidelinks]{hyperref}
\usepackage[utf8]{inputenc}
\usepackage[small]{caption}
\usepackage{graphicx}
\usepackage{amsmath}
\usepackage{booktabs}
\usepackage{amsthm}
\usepackage{amssymb}

\usepackage{cleveref,bm,multirow}
\usepackage{color}
\usepackage{xspace}
\usepackage{tikz}
\usepackage{tabularx}
\usepackage{pbox}
\usepackage{framed}
\usepackage[shortlabels]{enumitem}
\usepackage{thm-restate}

\usepackage{algorithm}
\usepackage{algorithmic}

\usepackage{natbib}
\usepackage{babel}
\usepackage{subcaption}
\usetikzlibrary{arrows.meta}
\usetikzlibrary{positioning}
\newcommand{\val}[2][]{\mu_{#1} ( #2 )} 
\newcommand{\mmax}[1]{\text{max} ( #1 )}
\newcommand{\mmin}[1]{\text{min} ( #1 )}

\renewcommand{\emptyset}{\varnothing}

\newcommand{\psn}[1]{\textsc{PSN} ( #1 )}
\newtheorem{innercustomthm}{}
\newenvironment{unnumberedthm}[1]
  {\renewcommand\theinnercustomthm{#1}\innercustomthm}
  {\endinnercustomthm}
\usepackage{cleveref}
\newtheorem{thm}{Theorem}[]
\newtheorem{prop}[thm]{Proposition}
\newtheorem{lem}[thm]{Lemma}

\theoremstyle{definition}

\allowdisplaybreaks

\title{Approximate Envy-Freeness in Graphical Cake Cutting}

\author{
Sheung Man Yuen
\qquad
Warut Suksompong\\National University of Singapore
}

\date{\vspace{-3ex}}

\begin{document}

\maketitle

\begin{abstract}
We study the problem of fairly allocating a divisible resource in the form of a graph, also known as graphical cake cutting.
Unlike for the canonical interval cake, a connected envy-free allocation is not guaranteed to exist for a graphical cake.
We focus on the existence and computation of connected allocations with low envy.
For general graphs, we show that there is always a $1/2$-additive-envy-free allocation and, if the agents' valuations are identical, a $(2+\epsilon)$-multiplicative-envy-free allocation for any $\epsilon > 0$.
In the case of star graphs, we obtain a multiplicative factor of $3+\epsilon$ for arbitrary valuations and $2$ for identical valuations.
We also derive guarantees when each agent can receive more than one connected piece.
All of our results come with efficient algorithms for computing the respective allocations.
\end{abstract}

\section{Introduction}

Cake cutting refers to the classic problem of fairly allocating a divisible resource such as land or advertising spaces---playfully modeled as a ``cake''---among agents who may have different values for different parts of the resource \citep{RobertsonWe98,Procaccia13}.
The most common fairness criteria in this literature are \emph{proportionality} and \emph{envy-freeness}.
Proportionality demands that if there are $n$~agents among whom the cake is divided, then every agent should receive at least $1/n$ of her value for the entire cake.
Envy-freeness, on the other hand, requires that no agent would rather have another agent's piece of cake than her own.
Early work in cake cutting established that a proportional and envy-free allocation that assigns to each agent a connected piece of cake always exists, regardless of the number of agents or their valuations over the cake \citep{DubinsSp61,Stromquist80,Su99}.

Although existence results such as the aforementioned guarantees indicate that a high level of fairness can be achieved in cake cutting, they rely on a typical assumption in the literature that the cake is represented by an interval.
This representation is appropriate when the resource corresponds to machine processing time or a single road, but becomes insufficient when one wishes to divide more complex resources such as networks.
For example, one may wish to divide road networks, railway networks, or power cable networks among different companies for the purpose of construction or maintenance.
In light of this observation, \citet{BeiSu21} introduced a more general model called \emph{graphical cake cutting}, wherein the cake can be represented by any connected graph.
With a graphical cake, a connected proportional allocation may no longer exist---see \Cref{fig:star} (left).
Nevertheless, these authors showed that more than half of the proportionality guarantee can be retained: any graphical cake admits a connected allocation such that every agent receives at least $1/(2n-1)$ of her entire value.

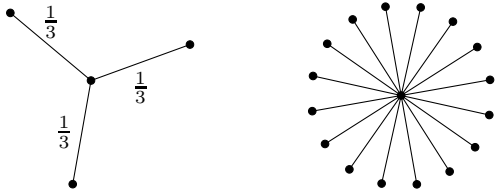
\begin{figure}[t]
\centering
  \begin{subfigure}[t]{.32\textwidth}
    \centering 
    \begin{tikzpicture}
    \usetikzlibrary{positioning}
    \node (v) at (0,0) {};
    \node at (20:1.4cm) (v1) {};
    \node at (-100:1.4cm) (v2) {};
    \node at (140:1.4cm) (v3) {};
    \draw [black, fill=black]
        (v) circle [radius=0.05] 
        (v1) circle [radius=0.05]
        (v2) circle [radius=0.05]
        (v3) circle [radius=0.05];
    \draw [-]
        (v1.center) -- (v.center)  node[pos=0.5][below]{$\frac{1}{3}$} 
        (v2.center) -- (v.center)  node[pos=0.5][left]{$\frac{1}{3}$} 
        (v3.center) -- (v.center)  node[pos=0.5][above]{$\frac{1}{3}$};
    \end{tikzpicture}
  \end{subfigure}
  \begin{subfigure}[t]{.32\textwidth}
    \centering 
    \begin{tikzpicture}
    \usetikzlibrary{positioning}
    \node (v) at (0,0) {};
    \node at (10:1.2cm) (v1) {};
    \node at (-12.5:1.2cm) (v2) {};
    \node at (-35:1.2cm) (v3) {};
    \node at (-57.5:1.2cm) (v4) {};
    \node at (-80:1.2cm) (v5) {};
    \node at (-102.5:1.2cm) (v6) {};
    \node at (-125:1.2cm) (v7) {};
    \node at (-147.5:1.2cm) (v8) {};
    \node at (-170:1.2cm) (v9) {};
    \node at (167.5:1.2cm) (v10) {};
    \node at (145:1.2cm) (v11) {};
    \node at (122.5:1.2cm) (v12) {};
    \node at (100:1.2cm) (v13) {};
    \node at (77.5:1.2cm) (v14) {};
    \node at (55:1.2cm) (v15) {};
    \node at (32.5:1.2cm) (v16) {};
    \draw [black, fill=black]
        (v) circle [radius=0.05] 
        (v1) circle [radius=0.05]
        (v2) circle [radius=0.05]
        (v3) circle [radius=0.05]
        (v4) circle [radius=0.05] 
        (v5) circle [radius=0.05]
        (v6) circle [radius=0.05]
        (v7) circle [radius=0.05]
        (v8) circle [radius=0.05] 
        (v9) circle [radius=0.05]
        (v10) circle [radius=0.05]
        (v11) circle [radius=0.05]
        (v12) circle [radius=0.05] 
        (v13) circle [radius=0.05]
        (v14) circle [radius=0.05]
        (v15) circle [radius=0.05]
        (v16) circle [radius=0.05];
    \draw [-]
        (v1.center) -- (v.center) -- (v9.center)
        (v2.center) -- (v.center) -- (v10.center)
        (v3.center) -- (v.center) -- (v11.center)
        (v4.center) -- (v.center) -- (v12.center)
        (v5.center) -- (v.center) -- (v13.center)
        (v6.center) -- (v.center) -- (v14.center)
        (v7.center) -- (v.center) -- (v15.center)
        (v8.center) -- (v.center) -- (v16.center);
    \end{tikzpicture}
  \end{subfigure}
\caption{(Left) A star graph with three edges of equal length. Two agents with identical valuations distributed uniformly over the three edges cannot each receive a connected piece worth at least $1/2$ of the whole cake at the same time.
(Right) A star graph with many edges to be divided between two agents. If sharing of vertices is disallowed, then the agent who does not receive the center vertex will be restricted to at most one edge, and will incur envy equal to almost the value of the entire cake.} \label{fig:star}
\end{figure}

The result of \citet{BeiSu21} demonstrates that approximate proportionality is attainable in graphical cake cutting.
However, the allocation that their algorithm produces may lead to high \emph{envy} between the agents.
In particular, while each agent~$i$ is guaranteed $1/(2n-1)$ of her value, it is possible that the algorithm assigns the remaining $(2n-2)/(2n-1)$ of the value to another agent~$j$ from $i$'s perspective, so that $i$ envies $j$ by almost the entire value of the cake (for large~$n$) when measured additively, and by a factor linear in $n$ when measured multiplicatively.
Note that envy-freeness is a much more stringent benchmark than proportionality---for instance, although there exists a simple protocol for computing a connected proportional allocation of an interval cake \citep{DubinsSp61}, no finite protocol can compute a connected envy-free allocation of it \citep{Stromquist08}, and even without the connectivity requirement, the only known envy-free protocol requires an enormous number of queries \citep{AzizMa16}.
The goal of our work is to investigate the existence of connected allocations of a graphical cake with low envy, as well as to design algorithms for computing such allocations.

\subsection{Our Results}
We assume that the cake is represented by the edges of a connected graph, and each edge can be subdivided into segments to be allocated to different agents.
Each agent is to receive a connected piece, though we will also briefly explore relaxations of this constraint towards the end. 
The whole cake must be allocated, and each agent's value for it is additive and normalized to~$1$; the value of each agent does not need to be uniform within each edge or across different edges.
Following \citet{BeiSu21}, we also assume that each vertex can be shared by multiple agents.\footnote{Without this assumption, one cannot obtain nontrivial envy-free guarantees---see the caption of \Cref{fig:star} (right) for a brief discussion.}
We consider both additive envy---for $\alpha\in [0,1]$, an allocation is \emph{$\alpha$-additive-EF} if no agent envies another agent by an \emph{amount} of more than $\alpha$---and multiplicative envy---for $\alpha\ge 1$, an allocation is \emph{$\alpha$-EF} if no agent envies another agent by a \emph{factor} of more than $\alpha$.
An $\alpha$-EF allocation is also $\left(\frac{\alpha-1}{\alpha+1}\right)$-additive-EF; we refer to \Cref{prop:relationships} for details.

In \Cref{sec:nonidenvalue}, we consider agents with (possibly) non-identical valuations.
We show that for any graph, there exists a $1/2$-additive-EF allocation, and such an allocation can be computed by iteratively allocating to each agent a share that other agents do not value too highly.
If the graph is a star, we present an algorithm that, for any $\epsilon > 0$, finds a $(3 + \epsilon)$-EF allocation (which is therefore nearly $1/2$-additive-EF as well) by allowing agents to repeatedly relinquish their current share for a higher-value share, and allocating the remaining shares by following certain rules.
Our two algorithms generalize ideas from algorithms for the interval cake by \citet{GoldbergHoSu20} and \citet{ArunachaleswaranBaRa19}, respectively.
We remark here that star graphs are of particular interest in graphical cake cutting because they constitute perhaps the most intuitive generalization of the well-studied interval cake, and therefore provide a natural platform for attempts to extend techniques and results from the interval-cake setting.\footnote{Star graphs (and path graphs) are also often studied in the context of indivisible items (see \Cref{subsec:related}). 
In graphical cake cutting, all path graphs are equivalent to the classic interval cake, which is why the role of star graphs is further highlighted.}

\begin{table}[t]
\centering
\begin{tabular}{c|cc}
& \textbf{General graphs} & \textbf{Star graphs} \\ \hline \rule{0cm}{3ex}
\textbf{Non-identical valuations} & $1/2$-additive-EF  (Thm.~\ref{thm:nonidenvalue-general}) & $(3+\epsilon)$-EF  (Thm.~\ref{thm:nonidenvalue-star}) \\ \rule{0cm}{3ex}
\textbf{Identical valuations} & $(2+\epsilon)$-EF  (Thm.~\ref{thm:idenvalue_general-2pluseps}) & $2$-EF  (Thm.~\ref{thm:idenvalue_star})
\end{tabular}
\caption{Summary of results in Sections \ref{sec:nonidenvalue} and \ref{sec:idenvalue}.} \label{tab:summary}
\end{table}

Next, in \Cref{sec:idenvalue}, we demonstrate how the bounds for non-identical valuations can be improved in the case of identical valuations; this case captures scenarios in which there is an objective valuation among agents.\footnote{We discuss further motivation for investigating this case at the beginning of \Cref{sec:idenvalue}.}
For arbitrary graphs, we devise an algorithm that computes a $(2+\epsilon)$-EF allocation (which is therefore nearly $1/3$-additive-EF).
Our algorithm is inspired by the work of \citet{ChuWuWa10} on partitioning edges of a graph (see \Cref{subsec:related}), and involves repeatedly adjusting the shares along a path from the minimum share to the maximum share so that the shares become more balanced in value.
For star graphs, we provide a simpler algorithm that returns a $2$-EF allocation using a bag-filling idea.
As we discuss after \Cref{prop:relationships}, an approximate proportionality result of \citet{BeiSu21} implies that both of our guarantees in this section are (essentially) tight.

Finally, in \Cref{sec:beyond}, we explore the fairness guarantees when each agent can receive more than one connected piece. We introduce the notion of \emph{path similarity number} to discuss the relationship between connected interval cake cutting and (non-connected) graphical cake cutting.

Our results in Sections \ref{sec:nonidenvalue} and \ref{sec:idenvalue} are summarized in \Cref{tab:summary}.
All of our algorithms can be implemented in the standard cake-cutting model of \citet{RobertsonWe98} in time polynomial in $n$, the size of the graph, and, if applicable, $1/\epsilon$.

\subsection{Further Related Work} \label{subsec:related}
Cake cutting is a topic of constant interest for researchers in mathematics, economics, and computer science alike.
For an overview of its intriguing history, we refer to the books by \citet{BramsTa96} and \citet{RobertsonWe98}, as well as the book chapter by \citet{Procaccia16}.

The cake-cutting literature traditionally assumes that the cake is given by an interval, and connectivity of the cake allocation is often desired in order to avoid giving agents a ``union of crumbs'' \citep{Stromquist80,Stromquist08,Su99,BeiChHu12,CechlarovaPi12,CechlarovaDoPi13,AumannDo15,ArunachaleswaranBaRa19,GoldbergHoSu20,SegalhaleviSu21,SegalhaleviSu23,BarmanKu22}.\footnote{While connectivity is the most frequently studied constraint in cake cutting, other constraints, such as geometric and separation constraints, have also been explored \citep{Suksompong21}.}
Besides \citet{BeiSu21}, a few authors have recently addressed the division of a graphical cake.
\citet{IgarashiZw23} focused on envy-freeness but made the crucial assumption that vertices cannot be shared between agents---as discussed in the caption of \Cref{fig:star} (right), with their assumption, one cannot obtain nontrivial guarantees even for star graphs and identical valuations.
\citet{DeligkasEiGa22} explored the complexity of deciding whether an envy-free allocation exists for a given instance (and, if so, finding one), both when vertices can and cannot be shared, but did not consider approximate envy-freeness.
\citet{ElkindSeSu21-Graph} investigated another fairness notion called maximin share fairness in graphical cake cutting.

Several recent papers have examined connectivity constraints in the allocation of \emph{indivisible items} represented by vertices of a graph \citep{BouveretCeEl17,IgarashiPe19,Suksompong19,LoncTr20,DeligkasEiGa21,BeiIgLu22,BeiLaLu24,BiloCaFl22,CaragiannisMiSh22,GahlawatZe23}.
In particular, \citet{CaragiannisMiSh22} assumed that some vertices can be shared by different agents; this assumption allowed them to circumvent the strong imbalance in the case of star graphs.
A number of authors considered the problem of dividing \emph{edges} of a graph, where, unlike in graphical cake cutting, each edge is treated as an indivisible object \citep{WuWaKu07,ChuWuWa10,ChuWuCh13}.

Another line of work also combines cake cutting and graphs, but the graph represents the acquaintance relation among agents \citep{AbebeKlPa17,BeiQiZh17,BeiSuWu20,GhalmeHuMa22,Tuckerfoltz23}.
Despite the superficial similarity, this model is very different from graphical cake cutting, and there are no implications between results on the two models.

\section{Preliminaries} \label{sec:prelim}
Let $G = (V, E)$ be a connected undirected graph representing the cake. 
Each edge $e \in E$, isomorphic to the interval $[0, 1]$ of the real numbers, is denoted by $e = [v_1, v_2] = [v_2, v_1]$ where $v_1, v_2 \in V$ are the endpoints of edge $e$. 
If $x_1$ and $x_2$ are points \emph{on} an edge, then the segment between them is denoted by $[x_1, x_2]$ or $[x_2, x_1]$---we sometimes call it an \emph{interval}. 
We shall restrict our attention to closed intervals only.

Two intervals of a cake $G$ are considered \emph{disjoint} if their intersection is a finite set of points.
A \emph{share} of a cake $G$ is a finite union of pairwise disjoint (closed) intervals of $G$---where the intervals may belong to different edges---such that it is connected, i.e., for any two points in the share, there exists a path between the two points that only traverses the share.\footnote{For brevity, we use the term ``share'' for what is sometimes referred to as a ``connected piece of cake''.} 
As with intervals, two shares of a cake $G$ are considered \emph{disjoint} if their intersection is a finite set of points.

Let $N = \{1, 2, \ldots, n\}$ be a set of $n \geq 2$ agents. 
Each agent $i \in N$ has a \emph{valuation function} (or \emph{utility function}) $\mu_i$, which maps each share of $G$ to a nonnegative real number. 
Following the cake-cutting literature \citep{Procaccia16}, we assume that each valuation function $\mu_i$ is 
\begin{itemize}
    \item \emph{normalized}: $\val[i]{G} = 1$, 
    \item \emph{divisible}: for each interval $[x_1, x_2]$ and $\lambda \in [0, 1]$, there exists a point $y \in [x_1, x_2]$ such that $\val[i]{[x_1, y]} = \lambda \cdot \val[i]{[x_1, x_2]}$, and 
    \item \emph{additive}: for any two disjoint intervals $I_1$ and $I_2$, we have $\val[i]{I_1 \cup I_2} = \val[i]{I_1} + \val[i]{I_2}$.
\end{itemize}
An \emph{instance} of graphical cake cutting consists of a graph $G$, a set of agents $N$, and valuation functions $( \mu_i )_{i \in N}$.

Given an instance, a \emph{partial allocation} of $G$ is a list $\mathcal{A} = (A_1, \ldots, A_n)$ of pairwise disjoint shares of the cake, where $A_i$ is agent $i$'s share---we say that $A_i$ is \emph{allocated} to agent~$i$. 
A share is \emph{unallocated} if its intersection with every agent's share is a finite set of points.
An \emph{allocation} of $G$ is a partial allocation such that every point in the cake is allocated to some agent. 
For a parameter $\alpha$, a partial allocation $\mathcal{A}$ is
\begin{itemize}
    \item \emph{$\alpha$-additive-EF} if $\val[i]{A_i} \geq \val[i]{A_j} - \alpha$ for all $i, j \in N$, 
    \item \emph{$\alpha$-EF} if $\val[i]{A_i} \geq \val[i]{A_j} / \alpha$ for all $i, j \in N$, 
    \item \emph{envy-free} if $\mathcal{A}$ is $1$-EF (or, equivalently, $0$-additive-EF), 
    \item \emph{$\alpha$-proportional} if $\val[i]{A_i} \geq 1 / (\alpha n)$ for all $i \in N$, and
    \item \emph{proportional} if $\mathcal{A}$ is $1$-proportional.
\end{itemize}

In order for algorithms to access the cake valuations, we assume that they can make \textsc{Cut} and \textsc{Eval} queries available in the standard model of \citet{RobertsonWe98}.

For completeness, we state the relationships between different fairness notions.

\begin{prop} \label{prop:relationships}
Let $\mathcal{A}$ be an allocation for $n \geq 2$ agents, and let $\alpha \geq 1$. 
\begin{itemize}
    \item If $\mathcal{A}$ is $\alpha$-EF, then it is $\left(\alpha - \frac{\alpha - 1}{n}\right)$-proportional.
    \item If $\mathcal{A}$ is $\alpha$-EF, then it is $\left( \frac{\alpha - 1}{\alpha + 1} \right)$-additive-EF.
    \item If $\mathcal{A}$ is $\alpha$-proportional, then it is $\left(1 - \frac{2}{\alpha n}\right)$-additive-EF.
\end{itemize}
\end{prop}

The proof of \Cref{prop:relationships} can be found in \Cref{app:relationships}.

\citet{BeiSu21} demonstrated that for every $n \geq 2$, there exists an instance in which no allocation is $\alpha$-proportional for any $\alpha < 2 - 1/n$, even for identical valuations and star graphs.
By \Cref{prop:relationships}, one cannot obtain a better guarantee than $2$-EF for such instances.

We now state a useful lemma about the existence of a share that has sufficiently high value for one agent and, at the same time, not exceedingly high value for other agents.

\begin{lem} \label{lem:divide}
Let $H$ be a connected subgraph of a graphical cake, and suppose that $H$ is worth $\beta_0$ to some agent in a subset $N' \subseteq N$. 
Then, for any positive $\beta \leq \beta_0$ and any vertex $r$ of $H$, there exists an algorithm, running in time polynomial in $n$ and the size of $H$, that finds a partition of $H$ into two (connected) shares such that the first share is worth at least $\beta$ to some agent in $N'$ and less than $2 \beta$ to every agent in~$N'$, and the second share contains the vertex $r$.
\end{lem}

\citet[Lemma~4.9]{BeiSu21} made this claim for the special case where all agents have the same value for~$H$ and no vertex~$r$ is specified.
We will use \Cref{lem:divide} as a subroutine in \textsc{IterativeDivide} (\Cref{alg:iterativedivide}) and \textsc{BalancePath} (\Cref{alg:balancepath}); \textsc{IterativeDivide} considers the case where different agents may have different values for $H$, while \textsc{BalancePath} requires the condition on the vertex $r$ in order to maintain connectivity along the minimum-maximum path.
We shall use \textsc{Divide}$(H, N', \beta, r)$ to denote the ordered pair of the two corresponding shares as described in the lemma.
The idea behind the proof of \Cref{lem:divide} is similar to that of the special case shown by \citet{BeiSu21}: we convert $H$ into a tree rooted at the vertex~$r$ by removing cycles iteratively---keeping the edges and duplicating the vertices if necessary---then traverse the tree from $r$ until a vertex $v$ with a subtree of an appropriate size is reached, and finally identify some connected subgraph of the subtree as the first share while assigning the remaining portion as the second share. 
The full details, including the pseudocode, are given in \Cref{app:divide}.

\section{Possibly Non-Identical Valuations} \label{sec:nonidenvalue}
In this section, we allow agents to have different valuations.
For arbitrary graphs, we present an algorithm that computes an approximately envy-free allocation of a graphical cake when measured additively. 
For the case where the graph is a star, we give an algorithm that finds an allocation wherein the envy is bounded by a multiplicative factor of roughly~$3$.

\subsection{General Graphs}
A priori, it is not even clear whether there exists a constant $\alpha < 1$ independent of $n$ such that an $\alpha$-additive-EF allocation always exists. 
We now describe the algorithm, \textsc{IterativeDivide} (\Cref{alg:iterativedivide}), which finds a $1/2$-additive-EF allocation for arbitrary graphs and non-identical valuations, using ideas similar to the algorithm by \citet{GoldbergHoSu20} for computing a $1/3$-additive-EF allocation of an interval cake.
Choose any arbitrary vertex~$r$ of $G$, and start with the entire graph $G$ and all agents in contention. 
If there is only one agent remaining, allocate the remaining graph to that agent. 
If the remaining graph is worth less than $\beta = 1/4$ to every remaining agent,\footnote{While the value of $\beta$ is the same for all iterations here, we write $\beta_i$ in the pseudocode because we will later consider a generalization in which $\beta$ can be different for different iterations.}  allocate an empty graph to any one of the remaining agents and remove this agent.
Otherwise, apply the algorithm \textsc{Divide} on the remaining graph and the remaining agents with threshold $\beta$. 
Allocate the first share to any agent who values that share at least $\beta$, and remove this agent along with her share. 
Repeat the procedure with the remaining graph until the whole graph is allocated. 
We claim that the resulting allocation is indeed $1/2$-additive-EF.

\begin{algorithm}[tb]
    \caption{\textsc{IterativeDivide}$(G, N)$.} \label{alg:iterativedivide}
    \textbf{Input}: Graph $G$, set of agents $N = \{1, \ldots, n\}$.\\
    \textbf{Output}: Allocation $(A_1, \ldots, A_n)$. \\
    \textbf{Initialization}: $r \leftarrow$ any vertex of $G$; \ $H_n \leftarrow G$; \ $N' \leftarrow N$. \\
    \vspace{-3.5mm}
    \begin{algorithmic}[1]
        \FOR{$i = 1, \ldots, n - 1$}
            \STATE $\beta_i \leftarrow 1/4$ \label{ln:beta}
            \IF{there exists $i' \in N'$ such that $\val[i']{H_n} \geq \beta_i$} \label{ln:if}
                \STATE $(H_i, H_n) \leftarrow \textsc{Divide}(H_n, N', \beta_i, r)$ \label{ln:divide}
                \STATE $i^* \leftarrow$ any agent in $N'$ who values $H_i$ at least $\beta_i$
            \ELSE
                \STATE $H_i \leftarrow \emptyset$
                \STATE $i^* \leftarrow$ any agent in $N'$
            \ENDIF
            \STATE $A_{i^*} \leftarrow H_i$
            \STATE $N' \leftarrow N' \setminus \{ i^* \}$
        \ENDFOR
        \STATE $A_j \leftarrow H_n$, where $j$ is the remaining agent in $N'$
        \STATE \textbf{return} $(A_1, \ldots, A_n)$
    \end{algorithmic}
\end{algorithm}

\begin{thm} \label{thm:nonidenvalue-general}
Given an instance of graphical cake cutting, there exists an algorithm that computes a $1/2$-additive-EF allocation in time polynomial in $n$ and the size of $G$.
\end{thm}

\begin{proof}
We claim that the algorithm \textsc{IterativeDivide} (\Cref{alg:iterativedivide}) satisfies the condition.
It is clear that the algorithm can be implemented in polynomial time; it remains to check that the allocation returned by the algorithm is $1/2$-additive-EF.
Let $i \in N$, and let $N_0 \subseteq N$ be the subset of agents who were allocated shares that correspond to the first share of some \textsc{Divide} procedure called by \textsc{IterativeDivide}. 
If $i \in N_0$, then agent~$i$ receives a share worth at least $\beta = 1/4$ to her by \Cref{lem:divide}, so every other agent receives a share worth at most $1 - 1/4 = 3/4$ to agent $i$, and agent $i$'s envy is at most $3/4 - 1/4 = 1/2$. 
Else, $i \notin N_0$, and every agent in~$N_0$ receives a share worth less than $2\beta = 1/2$ to agent $i$ by \Cref{lem:divide}, while every agent in $N \setminus N_0$ receives a share worth less than $\beta = 1/4 < 1/2$ to agent $i$, so agent $i$'s envy is again at most $1/2$.
\end{proof}

While an additive envy of $1/2$ can be seen as high, the left example of \Cref{fig:star} shows that an envy of $1/3$ is inevitable.
Moreover, even for an interval cake, the (roughly) $1/4$-additive approximation of \citet{BarmanKu22} is the current best as far as polynomial-time computability is concerned.

Although \textsc{IterativeDivide} guarantees that the envy between each pair of agents is at most $1/2$, it is possible that some agents receive an empty share from the algorithm. 
In the remainder of this paper, we present algorithms that find approximately envy-free allocations up to constant multiplicative factors for star graphs as well as for agents with identical valuations. 
Any such allocation ensures positive value for every agent and, by \Cref{prop:relationships}, is also approximately envy-free when measured additively.

\subsection{Star Graphs} \label{subsec:nonidenvalue-star}
The case of star graphs presents a natural generalization of the canonical interval cake and, as can be seen in \Cref{fig:star}, already highlights some of the challenges that graphical cake cutting poses.
For this class of graphs, we devise an algorithm that, for any constant $\epsilon > 0$, computes a $(3 + \epsilon)$-EF allocation in polynomial time. 
The algorithm consists of four phases. 
It starts with an empty partial allocation and finds a small star of ``stubs'' near the center vertex (Phase~1). 
It then repeatedly finds an unallocated share worth slightly more than some agent's share, and allows that agent to relinquish her existing share for this new share---care must be taken to ensure that other agents do not have too much value for this new share (Phase~2).
This new share could be a segment of an edge (Phase~2a) or a union of multiple complete edges (Phase~2b).
This phase is repeated until there are no more unallocated shares suitable for agents to trade with.
Finally, the unallocated shares are appended to the agents' existing shares (Phases 3 and 4).
See \Cref{fig:nonidenvalue-star} for an illustration of each phase.
We remark that Phases~2a and 3 of our algorithm are adapted from the algorithm of \citet{ArunachaleswaranBaRa19} for finding a $(2 + \epsilon)$-EF allocation of an interval cake. 

\begin{figure*}[t]
\centering
  \begin{subfigure}[b]{.24\textwidth}
    \centering 
    \begin{tikzpicture}[scale=0.86]
    \usetikzlibrary{positioning}
    \node[label=below:\small $v$] at (0,0) (v) {};
    \node at (180:1.5cm) (v1) {};
    \node at (120:1.5cm) (v2) {};
    \node at (60:1.5cm) (v3) {};
    \node at (0:1.5cm) (v4) {};
    \node at (-60:1.5cm) (v5) {};
    \node at (-120:1.5cm) (v6) {};
    \node at (180:0.4cm) (x1) {};
    \node at (120:0.4cm) (x2) {};
    \node at (60:0.4cm) (x3) {};
    \node at (0:0.4cm) (x4) {};
    \node at (-60:0.4cm) (x5) {};
    \node at (-120:0.4cm) (x6) {};
    \node at (60:0.2cm) (word2) {};
    \node at (30:1.3cm) (word1) {};
    \node at (30:1.6cm) (word1text) {\small $\leq \frac{\epsilon'}{m}$};
    \draw [black, fill=black]
        (v) circle [radius=0.05]
        (v1) circle [radius=0.05]
        (v2) circle [radius=0.05]
        (v3) circle [radius=0.05]
        (v4) circle [radius=0.05]
        (v5) circle [radius=0.05]
        (v6) circle [radius=0.05]
        (x1) circle [radius=0.05]
        (x2) circle [radius=0.05]
        (x3) circle [radius=0.05]
        (x4) circle [radius=0.05]
        (x5) circle [radius=0.05]
        (x6) circle [radius=0.05];
    \draw []
        (v.center) -- (v1.center)  node[pos=1.2]{\small $v_1$} node[pos=0.45, above]{\small $x_1$}
        (v.center) -- (v2.center)  node[pos=1.2]{\small $v_2$} node[pos=0.5, right]{\small $x_2$}
        (v.center) -- (v3.center)  node[pos=1.2]{\small $v_3$} node[pos=0.4, right]{\small $x_3$}
        (v.center) -- (v4.center)  node[pos=1.2]{\small $v_4$} node[pos=0.45, below]{\small $x_4$}
        (v.center) -- (v5.center)  node[pos=1.2]{\small $v_5$} node[pos=0.5, left]{\small $x_5$}
        (v.center) -- (v6.center)  node[pos=1.2]{\small $v_6$} node[pos=0.4, left]{\small $x_6$};
    \draw [->] 
        (word1.south) to [out=-90,in=0] (word2.east);
    \end{tikzpicture}
    \caption{Phase 1}
  \end{subfigure}
  \begin{subfigure}[b]{.23\textwidth}
    \centering 
    \begin{tikzpicture}[scale=0.86]
    \usetikzlibrary{positioning}
    \node at (0,0) (v) {};
    \node at (180:1.5cm) (v1) {};
    \node at (120:1.5cm) (v2) {};
    \node at (60:1.5cm) (v3) {};
    \node at (0:1.5cm) (v4) {};
    \node at (-60:1.5cm) (v5) {};
    \node at (-120:1.5cm) (v6) {};
    \node at (180:0.4cm) (x1) {};
    \node at (120:0.4cm) (x2) {};
    \node at (60:0.4cm) (x3) {};
    \node at (0:0.4cm) (x4) {};
    \node at (-60:0.4cm) (x5) {};
    \node at (-120:0.4cm) (x6) {};
    \node at (180:1.1cm) (y11) {};
    \node at (180:0.8cm) (y12) {};
    \node at (120:1.0cm) (y2) {};
    \draw [white]
        (v.center) -- (v1.center)  node[pos=1.2, black]{\small $v_1$}
        (v.center) -- (v2.center)  node[pos=1.2, black]{\small $v_2$}
        (v.center) -- (v3.center)  node[pos=1.2, black]{\small $v_3$}
        (v.center) -- (v4.center)  node[pos=1.2, black]{\small $v_4$}
        (v.center) -- (v5.center)  node[pos=1.2, black]{\small $v_5$}
        (v.center) -- (v6.center)  node[pos=1.2, black]{\small $v_6$};
    \draw []
        (x1.center) -- (v.center)
        (x2.center) -- (v.center)
        (x3.center) -- (v.center)
        (v4.center) -- (v.center)
        (v5.center) -- (v.center)
        (v6.center) -- (v.center);
    \draw [dotted, thick]
        (y11.center) -- (y12.center)
        (y2.center) -- (x2.center)
        (v3.center) -- (x3.center);
    \draw [|-|, double, blue]
        (v1.center) -- (y11.center)  node[pos=0.5, below, black]{\small $A_1$};
    \draw [|-|, double, brown]
        (y12.center) -- (x1.center)  node[pos=0.5, below, black]{\small $A_4$};
    \draw [|-|, double, teal]
        (v2.center) -- (y2.center)  node[pos=0.3, right, black]{\small $A_2$};
    \draw [|-, double, violet]
        (v6.center) -- (v.center);
    \draw [|-, double, violet]
        (v5.center) -- (v.center)  node[pos=0.5, right, black]{\small $A_3$};
    \draw [|-, double, violet]
        (v4.center) -- (v.center);
    \draw [black, fill=black]
        (v) circle [radius=0.05]
        (v1) circle [radius=0.05]
        (v2) circle [radius=0.05]
        (v3) circle [radius=0.05]
        (v4) circle [radius=0.05]
        (v5) circle [radius=0.05]
        (v6) circle [radius=0.05]
        (x1) circle [radius=0.05]
        (x2) circle [radius=0.05]
        (x3) circle [radius=0.05]
        (x4) circle [radius=0.05]
        (x5) circle [radius=0.05]
        (x6) circle [radius=0.05];
    \end{tikzpicture}
    \caption{Phase 2}
  \end{subfigure}
  \begin{subfigure}[b]{.26\textwidth}
    \centering
    \begin{tikzpicture}[scale=0.86]
    \usetikzlibrary{positioning}
    \node[label=below:\small $v_k$,label=above:\small (before)] at (0,0) (v1) {};
    \node at (0.4,0) (y11) {};
    \node at (1.0,0) (y12) {};
    \node at (1.3,0) (y13) {};
    \node at (1.8,0) (y14) {};
    \node at (2.1,0) (y15) {};
    \node at (2.8,0) (y16) {};
    \node[label=above:\small $x_k$] at (3.2,0) (x1) {};
    \node[label=below:\small $v$] at (3.5,0) (v) {};
    
    \node at (0.1, 0) (n1) {};
    \node at (1.15, 0) (n2) {};
    \node at (2.1, 0) (n3) {};
    \node at (3.1, 0) (n4) {};
    \node at (0.73, 0) (P1) {};
    \node at (1.55, 0) (P2) {};
    \node at (2.45, 0) (P3) {};
    
    \node[label=below:\small $v_k$,label=above:\small (after)] at (0,-2) (dv1) {};
    \node at (0.4,-2) (dy11) {};
    \node at (1.0,-2) (dy12) {};
    \node at (1.3,-2) (dy13) {};
    \node at (1.8,-2) (dy14) {};
    \node at (2.1,-2) (dy15) {};
    \node at (2.8,-2) (dy16) {};
    \node[label=above:\small $x_k$] at (3.2,-2) (dx1) {};
    \node[label=below:\small $v$] at (3.5,-2) (dv) {};
    
    \draw [black, fill=black]
        (v1) circle [radius=0.05]
        (x1) circle [radius=0.05]
        (v) circle [radius=0.05]
        (dv1) circle [radius=0.05] 
        (dx1) circle [radius=0.05]
        (dv) circle [radius=0.05];
    \draw [|-|, double, blue]
        (y11.center) -- (y12.center)  node[pos=0.5, below, black]{\small $A_1$};
    \draw [|-|, double, teal]
        (y13.center) -- (y14.center)  node[pos=0.5, below, black]{\small $A_2$};
    \draw [|-|, double, violet]
        (y15.center) -- (y16.center)  node[pos=0.5, below, black]{\small $A_3$};
    \draw [dotted, thick]
        (v1.center) -- (y11.center);
    \draw [dotted, thick]
        (y12.center) -- (y13.center);
    \draw [dotted, thick]
        (y14.center) -- (y15.center);
    \draw [dotted, thick]
        (y16.center) -- (x1.center);
    \draw []
        (x1.center) -- (v.center);   
    \draw [->] 
        (n1) to [out=60,in=120] (P1);    
    \draw [->] 
        (n2.north) to [out=135,in=45] (P1.north);    
    \draw [->] 
        (n3) to [out=135,in=45] (P2.north);    
    \draw [->] 
        (n4) to [out=135,in=45] (P3.north);
    \draw [|-|, double, blue]
        (dv1.center) -- (dy13.center)  node[pos=0.5, below, black]{\small $A_1$};
    \draw [|-|, double, teal]
        (dy13.center) -- (dy15.center)  node[pos=0.5, below, black]{\small $A_2$};
    \draw [|-|, double, violet]
        (dy15.center) -- (dx1.center)  node[pos=0.5, below, black]{\small $A_3$};
    \draw []
        (dx1.center) -- (dv.center);
    \draw[dashed]
        (y13.center) -- (dy13.center)
        (y15.center) -- (dy15.center);
    \end{tikzpicture}
    \caption{Phase 3}
  \end{subfigure}
  \begin{subfigure}[b]{.23\textwidth}
    \centering 
    \begin{tikzpicture}[scale=0.86]
    \usetikzlibrary{positioning}
    \node[label=below:\small $v$] at (0,0) (v) {};
    \node at (180:1.5cm) (v1) {};
    \node at (120:1.5cm) (v2) {};
    \node at (60:1.5cm) (v3) {};
    \node at (0:1.5cm) (v4) {};
    \node at (-60:1.5cm) (v5) {};
    \node at (-120:1.5cm) (v6) {};
    \node at (180:0.4cm) (x1) {};
    \node at (120:0.4cm) (x2) {};
    \node at (60:0.4cm) (x3) {};
    \node at (0:0.4cm) (x4) {};
    \node at (-60:0.4cm) (x5) {};
    \node at (-120:0.4cm) (x6) {};
    \node at (180:0.8cm) (y11) {};
    \draw [white]
        (v.center) -- (v1.center)  node[pos=1.2, black]{\small $v_1$}
        (v.center) -- (v2.center)  node[pos=1.2, black]{\small $v_2$}
        (v.center) -- (v3.center)  node[pos=1.2, black]{\small $v_3$}
        (v.center) -- (v4.center)  node[pos=1.2, black]{\small $v_4$}
        (v.center) -- (v5.center)  node[pos=1.2, black]{\small $v_5$}
        (v.center) -- (v6.center)  node[pos=1.2, black]{\small $v_6$};
    \draw []
        (v1.center) -- (x1.center)
        (v2.center) -- (x2.center)
        (v4.center) -- (v.center)
        (v5.center) -- (v.center)
        (v6.center) -- (v.center);
    \draw [ultra thick]
        (x1.center) -- (v.center)
        (x2.center) -- (v.center)
        (v3.center) -- (v.center)  node[pos=0.5, right]{\small $H$};
    \draw [|-|, double, blue]
        (v1.center) -- (y11.center);
    \draw [|-|, double, brown]
        (y11.center) -- (x1.center);
    \draw [|-|, double, teal]
        (v2.center) -- (x2.center);
    \draw [|-, double, violet]
        (v6.center) -- (v.center);
    \draw [|-, double, violet]
        (v5.center) -- (v.center);
    \draw [|-, double, violet]
        (v4.center) -- (v.center);  
    \draw [black, fill=black]
        (v) circle [radius=0.05]
        (v1) circle [radius=0.05]
        (v2) circle [radius=0.05]
        (v3) circle [radius=0.05]
        (v4) circle [radius=0.05]
        (v5) circle [radius=0.05]
        (v6) circle [radius=0.05]
        (x1) circle [radius=0.05]
        (x2) circle [radius=0.05]
        (x3) circle [radius=0.05]
        (x4) circle [radius=0.05]
        (x5) circle [radius=0.05]
        (x6) circle [radius=0.05];
    \end{tikzpicture}
    \caption{Phase 4}
  \end{subfigure}
\caption{(a) The points $x_k$ are found, where $[x_k, v]$ is worth at most $\epsilon'/m$ to every agent. (b) The unallocated intervals (dotted lines) are the ones to be considered in Phase 2a. (c)~The unallocated intervals (dotted lines) are appended leftwards in $v_k$'s direction, except for the one containing $v_k$ which is appended rightwards. (d) The remaining unallocated portion $H$ (bold lines) is a share connected by $v$.} \label{fig:nonidenvalue-star}
\end{figure*}
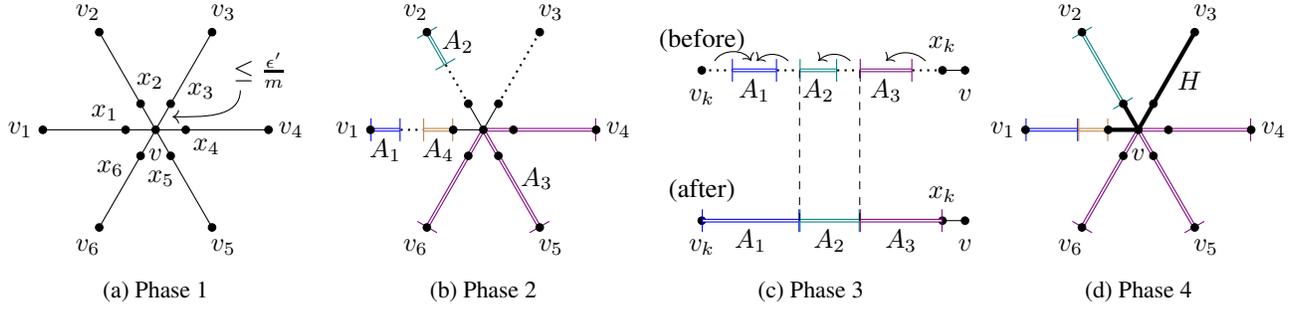

Let $G = (V, E)$ be a star graph centered at vertex $v$ with $m \geq 2$ edges. 
Label the other vertices $v_k$ and the edges $e_k = [v_k, v]$ for $k \in \{1,\dots, m\}$.
Fix any $\epsilon \in (0, 1)$.
\\\\
\noindent \textbf{Phase 1: Preparation. }
Define $\epsilon' = \frac{\epsilon}{16nm}$.
Initialize an empty partial allocation $\mathcal{A} = (A_1, \ldots, A_n)$. 
For each edge $e_k$, find a point $x_k \in [v_k, v]$ such that the segment $[x_k, v]$ is worth at most $\epsilon'/m$ to all agents. 
Define $e_k^1 = [v_k, x_k]$ and  $e_k^2 = [x_k, v]$ for $k \in \{1,\dots, m\}$, and let $E^1$ and $E^2$ be the sets containing all $e_k^1$'s and all $e_k^2$'s, respectively.
Note that $E^2$ is worth at most $\epsilon'$ to every agent.
\\\\
\noindent \textbf{Phase 2: Increase agents' shares incrementally. }
If there is a segment $e_k^1 \in E^1$ such that some unallocated interval within $e_k^1$ is worth at least $\val[i]{A_i} + \epsilon'$ to some agent $i$, \textbf{go to Phase 2a}. Otherwise, consider the segments in $E^1$ that are entirely unallocated. 
If the union of these segments is worth at least $\val[i]{A_i} + \epsilon'$ to some agent $i$, \textbf{go to Phase 2b}.
Otherwise, Phase 2 ends; \textbf{go to Phase 3}.
\begin{itemize}
    \item \textbf{Phase 2a: Allocate a subinterval of some $e_k^1$. }
    Pick an unallocated interval $I \subseteq e_k^1$ that is worth at least $\val[i]{A_i} + \epsilon'$ to some agent $i$, and assume without loss of generality that it cannot be extended in either direction without overlapping an allocated share or $e_k^2$.
    Suppose that $I = [a,b]$, where $a$ is closer to $v_k$ than $b$ is.
    If $a = v_k$, find the point $z \in I$ closest to $a$ such that $[a, z]$ is worth exactly $\val[i^*]{A_{i^*}} + \epsilon'$ to some agent~$i^*$, and let $A_{i^*} = [a, z]$, i.e., agent $i^*$ relinquishes her existing share for this new share.
    Else, $a \ne v_k$; find the point $z \in I$ closest to $b$ such that $[z, b]$ is worth exactly $\val[i^*]{A_{i^*}} + \epsilon'$ to some agent~$i^*$, and let $A_{i^*} = [z, b]$.
    \textbf{Repeat Phase 2.}
    \item \textbf{Phase 2b: Allocate multiple edges in $E$. } Let $K_0$ be the set of all indices~$k$ such that the entire segment $e_k^1$ is unallocated.
    Initialize $K = \emptyset$, and add the indices from $K_0$ to $K$ one by one until $\{ e_k^1 \mid k \in K \}$ is worth at least $\val[i^*]{A_{i^*}} + \epsilon'$ to some agent $i^*$. 
    Let $A_{i^*}$ be the union of $e_k$ over all $k \in K$, i.e., agent $i^*$ relinquishes her existing share for this new share. 
    Note that this new share is connected by the center vertex $v$. \textbf{Repeat Phase~2.}
\end{itemize}

\noindent \textbf{Phase 3: Append unallocated subintervals within $e_k^1$. }
Let $N_1\subseteq N$ consist of all agents who last received a subinterval of some $e_k^1$ via Phase 2a, and $N_2\subseteq N$ consist of all agents who last received two or more complete edges in~$E$ via Phase 2b (we will show later that, in fact, $N_1\cup N_2 = N$). 
For each $e_k^1$ of which some agent from $N_1$ is allocated a subinterval, and for each unallocated interval $I = [a, b] \subseteq e_k^1$ that cannot be extended in either direction without overlapping an allocated share or $E^2$, where $a$ is closer to $v_k$ than $b$ is, append $I$ to the share of the agent who is allocated the point $a$ (i.e., append towards $v_k$'s direction).
The only time this is not possible is when $a = v_k$, in which case we append to the share of the agent who is allocated the point~$b$. 
\\\\
\noindent \textbf{Phase 4: Append $H$. }
Consider the remaining unallocated portion $H$ of the graph.
Note that for each $k \in \{1, \ldots, m\}$, we have $H \cap e_k = \{v\}$ or $e_k^2$ or $e_k$---this means that $H$ is connected by the center vertex $v$.
If $N_2$ is nonempty, append $H$ to the share of an arbitrary agent in $N_2$.
Else, if some segment $e_k^1$ is allocated to at least two agents, give $H$ to the agent who has been allocated the point~$x_k$.
Otherwise, we know that every agent is allocated exactly one segment in $E^1$---give $H$ to the agent who traded her share last in Phase~2 (in particular, Phase~2a).
\smallskip

We claim that this algorithm yields a $(3 + \epsilon)$-EF allocation.
By \Cref{prop:relationships}, such an allocation is roughly $1/2$-additive-EF and (by taking $\epsilon = 1/n$) $3$-proportional as well.

\begin{thm} \label{thm:nonidenvalue-star}
Given an instance of graphical cake cutting consisting of a star graph with $m$ edges, there exists an algorithm that, for any $\epsilon > 0$, computes a $(3 + \epsilon)$-EF allocation in time polynomial in $n$, $m$, and $1/\epsilon$.
\end{thm}

For the sake of exposition, we shall introduce notations to differentiate the partial allocations at different stages of the algorithm.
For an integer $t\ge 0$, let $\mathcal{P}^t = (P_1^t, \ldots, P_n^t)$ be the partial allocation after $t$ iterations of Phase 2, and let $\mathcal{P} = (P_1, \ldots, P_n)$ be the partial allocation at the start of Phase 3 (we show in \Cref{lem:valid} that it is a valid partial allocation).
For any share~$P$, let $\widehat{P} = \bigcup_{k=1}^m (P \cap e_k^1)$. 
The final allocation (i.e., after Phase~4) shall be denoted $\mathcal{A} = (A_1, \ldots, A_n)$. 
We establish the approximate envy-freeness of $\mathcal{A}$ via a series of intermediate results.

\begin{lem} \label{lem:valid}
$\mathcal{P}$ is a valid partial allocation, and is equal to $\mathcal{P}^t$ for some $t \leq \frac{16n^2m}{\epsilon}$.
\end{lem}

\begin{proof}
To show that $\mathcal{P}$ is a valid partial allocation, we prove by induction that $\mathcal{P}^t$ is a valid partial allocation for every~$t$. 
In particular, we check that each $P_i^t$ is connected, and the agents' shares in $\mathcal{P}^t$ are pairwise disjoint. 
Clearly, the empty partial allocation, $\mathcal{P}^0$, is a valid partial allocation. 
Now, assume that $\mathcal{P}^t$ is a valid partial allocation; we will prove the validity of $\mathcal{P}^{t+1}$. 
At the $(t+1)^\text{th}$ iteration of Phase~2, some agent $i^*$ trades her share in either Phase 2a or 2b, while all other agents' shares remain unchanged, so we only need to check that the share of agent~$i^*$ is connected and disjoint from other agents' shares.
\begin{itemize}
    \item If agent $i^*$ trades her share in Phase 2a, then she receives a connected subinterval of some $e_k^1$; furthermore, this subinterval is disjoint from other agents' shares since it is a subset of some unallocated interval $I = [a, b]$.
    \item If agent $i^*$ trades her share in Phase 2b, then she receives a collection of edges $e_k$ which are connected by the vertex $v$; furthermore, for any $k \in \{1,\dots, m\}$, if $e_k^1$ is unallocated, then $e_k^2$ is unallocated as well, so the $e_k$'s received by agent~$i^*$ are not allocated to any other agent.
\end{itemize}
Hence, $\mathcal{P}^{t+1}$ is a valid partial allocation. This completes the induction.

In each iteration of Phase 2, some agent increases the value of her share by at least~$\epsilon'$. 
Since the value of each agent's share starts from~$0$ and cannot exceed~$1$, the total number of increments is at most $1/\epsilon'$ for each agent. 
As there are $n$ agents, the total number of iterations of Phase 2 is at most $n/\epsilon' = 16n^2m/\epsilon$.
\end{proof}

\begin{lem} \label{lem:envy}
Fix any $i \in N$.
\begin{itemize}
    \item For any $j \in N_1$, we have $\val[i]{P_j} \leq \val[i]{P_i} + \epsilon'$.
    \item For any $j \in N_2$, we have $\val[i]{\widehat{P}_j} \leq 2(\val[i]{P_i} + \epsilon')$.
\end{itemize}
\end{lem}

\begin{proof}
We prove that the statements are true for each $\mathcal{P}^t$ by induction on~$t$. 
Let $N_1^t$ and $N_2^t$ be the sets of agents whose share in $\mathcal{P}^t$ was last obtained via Phase~2a and Phase~2b, respectively. 
Note that $N_1^0 = N_2^0 = \emptyset$, and if agent $i^*$ obtains her share via Phase~2a in the $t^\text{th}$ iteration of Phase~2, then $N_1^t = N_1^{t-1} \cup \{i^* \}$ and $N_2^t = N_2^{t-1} \setminus \{i^* \}$; an analogous statement holds if $i^*$ obtains her share via Phase~2b. 
For the induction, we need to prove the following two statements.
\begin{itemize}
    \item For any $j \in N_1^t$, we have $\val[i]{P_j^t} \leq \val[i]{P_i^t} + \epsilon'$.
    \item For any $j \in N_2^t$, we have $\val[i]{\widehat{P}_j^t} \leq 2(\val[i]{P_i^t} + \epsilon')$.
\end{itemize}
The statements are clearly true for the empty partial allocation $\mathcal{P}^0$, as all shares have zero value.
Now, assume that the statements are true for $\mathcal{P}^t$; we shall prove the same for $\mathcal{P}^{t+1}$. 
Only the share of one agent~$i^*$ has changed, so we can focus on the case where either $i = i^*$ or $j = i^*$.
Since the statements trivially hold for $i = j$, we may assume that $i \neq j$.
\begin{itemize}
    \item If $i = i^*$, then we have $P_j^{t+1} = P_j^t$ and $\val[i]{P_i^{t+1}} \geq \val[i]{P_i^t}$, so both statements hold for $t + 1$.
    \item If $j = i^*$, then agent $j$ trades her share in either Phase 2a or Phase 2b in the $(t + 1)^\text{th}$ iteration of Phase~2.
    \begin{itemize}
        \item If $j \in N_1^{t+1}$, then $j$ trades in Phase~2a.
        By our procedure in Phase~2a, agent~$j$'s share is not worth more than $\val[i]{P_i^t} + \epsilon' = \val[i]{P_i^{t+1}} + \epsilon'$ to agent~$i$; otherwise agent $i$ would have gotten a strict subinterval of $P_j^{t+1}$ instead.
        Therefore, the first statement holds for $t+1$.
        \item If $j \in N_2^{t+1}$, then $j$ trades in Phase~2b. 
        Recall that the set~$K$ in Phase~2b was formed by adding indices from $K_0$ one by one. 
        Let $K_1$ be the subset of~$K$ without the last index added. 
        Then each of $\{ e_k^1 \mid k \in K_1\}$ and $\{ e_k^1 \mid k \in K \setminus K_1 \}$ is worth less than $\val[i]{P_i^t} + \epsilon' = \val[i]{P_i^{t+1}} + \epsilon'$ to agent $i$, so their union, $\widehat{P}_j^{t+1}$, is worth less than $2(\val[i]{P_i^{t+1}} + \epsilon')$ to agent $i$.
        Therefore, the second statement holds for $t+1$.
    \end{itemize}
\end{itemize}
This completes the induction. By \Cref{lem:valid}, $\mathcal{P} = \mathcal{P}^t$ for some $t$, and so the statements in \Cref{lem:envy} hold.
\end{proof}

\begin{lem} \label{lem:positive}
Every agent receives a share in $\mathcal{P}$ worth at least $\frac{1}{4nm}$ to her.
\end{lem}

\begin{proof}
Suppose by way of contradiction that some agent $i$ receives a share worth less than $\frac{1}{4nm}$ to her. Since $\epsilon < 1$ and $\epsilon' = \frac{\epsilon}{16nm}$, we have $\val[i]{P_i} + \epsilon' < \frac{1}{4nm} + \epsilon' < \frac{5}{16nm}$. Every part of the graph $G$ can be classified into one of the following three cases.
\begin{itemize}
    \item \textbf{Case 1: Within $E^1$ and within an agent's share.} \\
    By \Cref{lem:envy}, for each~$j$, the value of $\widehat{P}_j$ is at most $2(\val[i]{P_i} + \epsilon') < \frac{5}{8nm}$ to agent~$i$. (If $j\not\in N_1\cup N_2$, then $P_j$ is empty.)
    As there are $n$ agents, the union of these $\widehat{P}_j$'s is worth at most $n\cdot \left( \frac{5}{8nm} \right) < \frac{5}{8}$ to agent $i$.
    \item \textbf{Case 2: Within an unallocated subinterval of some segment in $E^1$.} \\
    For each $k \in \{1,\dots, m\}$, if $n_k$ agents are allocated some subinterval of segment $e_k^1 \in E^1$, then there are at most $n_k + 1$ unallocated subintervals on the same segment.
    Therefore, altogether there are at most $n + m$ unallocated subintervals within all the segments in $E^1$. 
    Since Phase~2 terminated, each of these subintervals is worth less than $\val[i]{P_i} + \epsilon' < \frac{5}{16nm}$ to agent $i$. 
    Therefore, the total value of these subintervals to agent~$i$ is less than $(n + m) \left(\frac{5}{16nm}\right) \leq \frac{5}{16}$, where the inequality holds because $n,m\ge 2$.
    \item \textbf{Case 3: Within $E^2$.} \\
    By definition, $E^2$ is worth at most $\epsilon' < \frac{1}{16}$ to agent~$i$.
\end{itemize}
The whole cake is thus worth less than $\frac{5}{8} + \frac{5}{16} + \frac{1}{16} = 1$ to agent~$i$, which is a contradiction.
\end{proof}

\Cref{lem:positive} implies that no agent receives an empty share in $\mathcal{P}$, that is, $N_1\cup N_2 = N$.
With this lemma in hand, we are now ready to prove \Cref{thm:nonidenvalue-star}.

\begin{proof}[Proof of \Cref{thm:nonidenvalue-star}]
Without loss of generality, we may assume that $\epsilon \in (0, 1)$. 
The running time claim holds because each iteration of each phase runs in time polynomial in $n$ and $m$, and the number of iterations of Phase 2 is polynomial in $n$, $m$, and $1/\epsilon$ by \Cref{lem:valid}.

Fix any $i, j \in N$; we shall first show that $A_j$ is worth at most $3\val[i]{P_i} + 4\epsilon'$ to agent~$i$. 
To this end, we consider three cases for $j$.
Recall the definition of $H$ from Phase~4 of the algorithm.
\begin{itemize}
    \item \textbf{Case 1: $j \in N_1$ and $A_j$ does not contain $H$.} \\
    By \Cref{lem:envy}, we have $\val[i]{P_j} \leq \val[i]{P_i} + \epsilon'$.
    Note that at the start of Phase 3, every unallocated subinterval of any $e_k^1$ is worth less than $\val[i]{P_i} + \epsilon'$ to agent $i$; otherwise Phase 2 would have continued. 
    Since agent~$j$ is allocated at most two such subintervals, we have $\val[i]{A_j} \leq 3(\val[i]{P_i} + \epsilon') \leq 3\val[i]{P_i} + 4\epsilon'$. \smallskip
    
    \item \textbf{Case 2: $j \in N_1$ and $A_j$ contains $H$.} \\
    Since $j \in N_1$, $P_j$ is a subinterval of some $e_k^1$. 
    Let $P_j = [a, b]$, where $a$ is closer to~$v_k$ than $b$ is. 
    We claim that in this case, unlike in Case~1, agent $j$ receives at most one unallocated subinterval of $e_k^1$ in Phase 3.
    Note that since $A_j$ contains $H$, no agent held a subinterval of $[b,x_k]$ at the start of Phase~3.
    \begin{itemize}
        \item If $e_k^1$ is allocated to at least two agents during Phase~3, then some other agent held a subinterval of $[v_k,a]$ at the start of Phase~3---let such a subinterval closest to $a$ be $[y,z]$, where $v_k\le y < z\le a$.
        The unallocated interval $[z, a]$ at the start of Phase~3 (if it is nonempty) is appended to the share of the agent who held $[y, z]$ at the start of Phase 3. 
        Thus, agent~$j$ receives only the unallocated subinterval $[b, x_k]$ of $e_k^1$ in Phase~3.
        \item If $e_k^1$ is allocated only to agent~$j$ during Phase 3, then since $A_j$ contains $H$, by the description of how $H$ is allocated in Phase~4, agent~$j$ was the last agent who traded her share in Phase~2, in particular, Phase~2a.
        This means that either $e_k^1$ was entirely unallocated just before agent~$j$ received a share from it, or it was allocated only to $j$ at that point.
        Hence, $j$'s share is of the form $[a, b] = [v_k, b]$ or $[a, b] = [a, x_k]$.
        It follows that $j$ receives only one unallocated subinterval $[b, x_k]$ or $[v_k, a]$ during Phase~3.
    \end{itemize}
    In total, $A_j$ consists of $P_j$ (worth at most $\val[i]{P_i} + \epsilon'$ to agent~$i$, by \Cref{lem:envy}), at most one unallocated subinterval of some $e_k^1$ (worth at most $\val[i]{P_i} + \epsilon'$ to agent~$i$, as in Case 1), and $H$. 
    Now, $H$ is a union of the unallocated segments in $E^1$ and a subset of $E^2$. The unallocated segments in $E^1$ together are worth less than $\val[i]{P_i} + \epsilon'$ to agent $i$---otherwise Phase~2 (in particular, Phase~2b) would have continued---and any subset of $E^2$ is worth at most $\epsilon'$ to agent $i$, so $H$ is worth at most $\val[i]{P_i} + 2\epsilon'$ to agent $i$.
    As a consequence, $A_j$ is worth at most $3\val[i]{P_i} + 4\epsilon'$ to agent $i$. \smallskip
    
    \item \textbf{Case 3: $j \in N_2$.} \\
    By \Cref{lem:envy}, we have $\val[i]{\widehat{P}_j} \leq 2(\val[i]{P_i} + \epsilon')$. 
    The remaining portion $A_j \setminus \widehat{P}_j$ is a subset of the union of the unallocated segments in $E^1$ and the segments of $E^2$, which is worth at most $\val[i]{P_i} + 2\epsilon'$ to agent $i$ as detailed in the last paragraph of Case~2.
    This gives $\val[i]{A_j} \leq 3\val[i]{P_i} + 4\epsilon'$.
\end{itemize}
In summary, we have $\val[i]{A_j} \leq 3\val[i]{P_i} + 4\epsilon'$ in all cases.
Now, $\epsilon' = \frac{\epsilon}{16nm}$ by definition and $\val[i]{P_i} \geq \frac{1}{4nm}$ by \Cref{lem:positive}, which implies that $4\epsilon' \leq \epsilon \val[i]{P_i}$. 
It follows that 
\begin{align*}
\val[i]{A_j} &\leq 3\val[i]{P_i} + 4\epsilon' 
\leq 3\val[i]{P_i} + \epsilon \val[i]{P_i} 
= (3 + \epsilon)\val[i]{P_i} 
\leq (3 + \epsilon)\val[i]{A_i}.
\end{align*}
Since $i, j \in N$ were arbitrarily selected, the allocation~$\mathcal{A}$ is $(3 + \epsilon)$-EF, as desired.
\end{proof}

\section{Identical Valuations} \label{sec:idenvalue}
In this section, we focus on the case where the valuation functions of all agents are identical. 
While this case is uninteresting for interval cake cutting since a fully envy-free allocation can be trivially found, it becomes highly nontrivial when graphs are involved (see, for example, the left of \Cref{fig:star}). 
Indeed, a number of works on dividing edges or vertices of a graph can be interpreted as dealing with the identical-valuation setting \citep{WuWaKu07,ChuWuWa10,CaragiannisMiSh22}.
Moreover, this setting captures scenarios where there is an objective measure across agents, for example, when a town wants to divide the responsibility of maintaining its streets among contractors based on the lengths or numbers of residents on the streets.

As we mentioned in \Cref{sec:prelim}, an $\alpha$-EF allocation is not guaranteed to exist for any $\alpha < 2$, even with identical valuations and star graphs. 
We will show in this section that, for arbitrary graphs and any $\epsilon > 0$, it is possible to find an allocation that is $(2+\epsilon)$-EF, which means that the approximation factor of~$2$ is essentially tight.

To this end, we first discuss how we can find a $4$-EF allocation using a variation of the \textsc{IterativeDivide} algorithm that we saw in \Cref{sec:nonidenvalue}. 
This $4$-EF allocation will later be used as an input to an algorithm that computes a $(2 + \epsilon)$-EF allocation.
For star graphs, we also describe a simpler method for computing a $2$-EF allocation. 

Let us denote by~$\mu$ the common valuation function of the agents, and define $\mmax{\mathcal{A}} = \max_{i\in N} \val{A_i}$ and $\mmin{\mathcal{A}} = \min_{i\in N} \val{A_i}$ for any allocation $\mathcal{A} = (A_1, \ldots, A_n)$. 

\subsection{\texorpdfstring{$4$}{4}-EF}
In \textsc{IterativeDivide}, we used the threshold $\beta = 1/4$ in every call to \textsc{Divide} so as to allocate a share worth at least $1/4$ to some agent, which results in a $1/2$-additive-EF allocation.
Even with identical valuations, each iteration of \textsc{Divide} is unpredictable in the sense that the recipient could receive a share worth anywhere between $\beta$ and $2\beta$.
If $\beta$ is chosen to be more than $1/(2n-2)$ and the first $n-1$ agents all take shares of value close to $2\beta$, then the last agent will be left effectively empty-handed. 
In contrast, if $\beta$ is chosen to be at most $1/(2n-2)$ and the first $n-1$ agents all take shares of value only~$\beta$, then the last agent will receive a share of value at least $1/2$, which leads to an envy factor linear in $n$.

To resolve this problem, let us consider using an \emph{adaptive} threshold that takes the values of the previous shares into account. 
If the previous agents took large shares, then the threshold $\beta$ is reduced appropriately for the current agent, and vice versa. 
Without loss of generality, assume that for each $i\in\{1,\dots,n-1\}$, agent $i$ is the one who takes the first share generated by the $i^\text{th}$ iteration of \textsc{Divide}. 
By choosing $\beta_i = \frac{1}{2} (\frac{2i}{2n-1} - \sum_{j=1}^{i-1} \val{A_j})$ to be the threshold for the $i^\text{th}$ iteration of \textsc{Divide}, we claim that the resulting allocation is $4$-EF.\footnote{Even better, the allocation is actually $(4 - 1/2^{n-3})$-EF, as shown in the proof of \Cref{thm:idenvalue_general-4ef}.} 
Along the way, we shall see that the allocation is also $(2 - 1/n)$-proportional. 

\begin{thm} \label{thm:idenvalue_general-4ef}
Given an instance of graphical cake cutting consisting of $n$ agents with identical valuations, there exists an algorithm that computes a $4$-EF and $(2 - \frac{1}{n})$-proportional allocation in time polynomial in $n$ and the size of $G$.
\end{thm}

\begin{proof}
We consider the algorithm \textsc{IterativeDivide} (\Cref{alg:iterativedivide})---without loss of generality, assume that the shares are allocated to the agents in ascending order of indices, i.e., $i^* = i$ for all $i \in \{ 1, \ldots, n-1 \}$ in the algorithm---and substitute $\beta_i$ with $\frac{1}{2} (\frac{2i}{2n-1} - \sum_{j=1}^{i-1} \val{A_j})$ (instead of $1/4$). 
We claim that this algorithm satisfies the condition of the theorem.
It is clear that the running time is polynomial in $n$ and the size of $G$, so it remains to check that the envy-freeness and proportionality claims are valid.
Let $(A_1, \ldots, A_n) = \textsc{IterativeDivide}(G, N)$, and let $\xi = \frac{1}{2n - 1}$. 
By induction, we shall prove the following statements for $i \in \{ 1, \ldots, n-1 \}$:
\begin{itemize}
    \item[(i)] $\; \displaystyle \xi \leq \val{A_i} < (4 - 2^{-(i-2)})\xi$;
    \item[(ii)] $\; \displaystyle (2i -2 + 2^{-(i-1)})\xi \leq \sum_{j=1}^{i} \val{A_j} < 2i\xi$.
\end{itemize}

For the base case $i = 1$, we have $\beta_1 = \xi \leq 1 = \val{G}$, which by \Cref{lem:divide} means that $\xi \leq \val{A_1} < 2\xi$, proving the two statements together.

For the inductive step, fix $i \in \{2, \ldots, n - 1\}$, and assume that the two statements hold for $i - 1$. 
We first check that the \textbf{if}-condition in Line~\ref{ln:if} of \textsc{IterativeDivide} is satisfied. 
To this end, we have to check that $\beta_i$ is positive and is at most the value of the remaining graph $1 - \sum_{j=1}^{i-1} \val{A_j}$. It follows from (ii) for $i - 1$ that $\sum_{j=1}^{i-1} \val{A_j} < 2(i-1)\xi$, so the former claim is true by $\beta_i = \frac{1}{2} (2i\xi - \sum_{j=1}^{i-1} \val{A_j}) > \xi > 0$ and the latter claim is true by $2\beta_i + \sum_{j=1}^{i-1} \val{A_j} = 2i\xi < 1$. The call to \textsc{Divide} in the following line is hence valid.
Therefore, by \Cref{lem:divide}, we have $\beta_i \leq \val{A_i} < 2\beta_i$.
We now prove the two inductive statements for~$i$.
\begin{itemize}
    \item By definition, $2\beta_i = 2i\xi - \sum_{j=1}^{i-1} \val{A_j}$, and by the inductive hypothesis, 
    \[
    (2(i-1)-2 + 2^{-(i-2)})\xi \leq \sum_{j=1}^{i-1} \val{A_j} < 2(i-1)\xi.
    \]
    Putting these together gives $2\xi < 2\beta_i \leq (4 - 2^{-(i-2)})\xi$. Thus, 
    \[
        \xi < \beta_i \leq \val{A_i} < 2\beta_i \leq (4 - 2^{-(i-2)})\xi,
    \]
    which proves (i).
    \item Now, subtracting $\beta_i \leq \val{A_i} < 2\beta_i$ from $2\beta_i$ gives $0 < 2\beta_i - \val{A_i} \leq \beta_i$. Then,
    \begin{align*}
        2i\xi - \sum_{j=1}^i \val{A_j} &= \left( 2i\xi - \sum_{j=1}^{i-1} \val{A_j} \right) - \val{A_i} 
        = 2\beta_i - \val{A_i};
    \end{align*}
    combining this with the previous inequality yields $0 < 2i\xi - \sum_{j=1}^i \val{A_j} \leq \beta_i$. Subtracting this from $2i\xi$ gives $2i\xi - \beta_i \leq \sum_{j=1}^i \val{A_j} < 2i\xi$. 
    Finally, combining this with the statement from the previous bullet point that $2\beta_i \leq (4 - 2^{-(i-2)})\xi$, that is, $\beta_i \leq (2 - 2^{-(i-1)})\xi$, we get 
    \[
        2i\xi - (2 - 2^{-(i-1)})\xi \leq \sum_{j=1}^i \val{A_j} < 2i\xi,
    \]
    which proves (ii).
\end{itemize}
This completes the induction.

By (ii) for $i = n-1$, we have
\[
(2(n-1)-2 + 2^{-(n-2)})\xi \leq \sum_{j=1}^{n-1} \val{A_j} < 2(n-1)\xi.
\]
Combining this with 
\[
\val{A_n} = 1 - \sum_{j=1}^{n-1} \val{A_j} =  (2n-1)\xi - \sum_{j=1}^{n-1} \val{A_j},
\]
we get $\xi < \val{A_n} \leq (3 - 2^{-(n-2)})\xi$. 
Together with (i), we see that the minimum value across all $\val{A_i}$'s is at least~$\xi$ and the maximum value is at most $(4 - 2^{-(n-3)})\xi$. This shows that the allocation is $(4 - 2^{-(n-3)})$-EF, which is also $4$-EF. 
Additionally, we have $\val{A_i} \geq \xi = \frac{1}{2n-1} = \frac{1}{(2-1/n)n}$ for all $i$, and so the allocation is $(2 - \frac{1}{n})$-proportional.
\end{proof}

As mentioned in \Cref{sec:prelim}, an $\alpha$-proportional allocation is not guaranteed to exist for any $\alpha < 2 - 1/n$, so the algorithm in \Cref{thm:idenvalue_general-4ef} attains the optimal proportionality approximation. 
In terms of envy-freeness, one may consider adjusting the values of $\beta_i$ in \textsc{IterativeDivide} so as to obtain an allocation with a better approximation factor than~$4$. 
While this might be possible, it seems unlikely that this approach could lead to $2$-EF, given that the guarantee from \Cref{lem:divide} already has a multiplicative gap of~$2$.
This motivates us to devise another algorithm that reduces the factor to arbitrarily close to~$2$.

\subsection{\texorpdfstring{$(2 + \epsilon)$}{(2 + eps)}-EF}
Let us first define a \emph{minimum-maximum path} of an allocation $\mathcal{A} = (A_1, \ldots, A_n)$ as a list $\mathcal{P} = (P_1, \ldots, P_d)$ (where $d\le n$) satisfying the following conditions:
\begin{itemize}
    \item For each $i \in \{1,\dots, d\}$, $P_i = A_j$ for some $j \in \{1,\dots, n\}$,
    \item $P_i \neq P_j$ for $1 \leq i < j \leq d$,
    \item For each $i \in \{1,\dots, d-1\}$, there exists at least one point belonging to both $P_i$ and $P_{i+1}$, and
    \item $\val{P_1} = \mmin{\mathcal{A}}$ and $\val{P_d} = \mmax{\mathcal{A}}$.
\end{itemize}
Intuitively, a minimum-maximum path is a list of shares that chains a minimum-valued one to a maximum-valued one in the underlying graph without crossing any share more than once.
Such a list can be found in polynomial time: we locate a minimum-valued share and a maximum-valued share of $\mathcal{A}$, find a path through the graph that connects both shares, and identify the shares corresponding to this path. 
If some share $A_i$ appears more than once, we repeatedly remove the part between the two occurrences of $A_i$ (including one of these occurrences). 
Let us use \textsc{MinMaxPath}$(\mathcal{A})$ to denote an arbitrary minimum-maximum path of $\mathcal{A}$.

We now describe the algorithm, \textsc{RecursiveBalance} (\Cref{alg:recursivebalance}), which finds a $(2 + \epsilon)$-EF allocation for any given $\epsilon > 0$. 
We employ similar ideas as the ones used by \citet{ChuWuWa10}---in their work, they partition the (indivisible) edges of a graph, whereas we have to account for the divisibility of the edges. 
Assume without loss of generality that $\epsilon \in (0, 1)$. 
Given an allocation $\mathcal{A}$, \textsc{RecursiveBalance} repeatedly replaces $\mathcal{A}$ with the allocation \textsc{Balance}$(\mathcal{A}, \epsilon)$, then terminates when $\mathcal{A}$ is $(2 + \epsilon)$-EF. 
The algorithm \textsc{Balance} (\Cref{alg:balance}) finds a minimum-maximum path $\mathcal{P}$ of $\mathcal{A}$, and replaces the shares in $\mathcal{A}$ that appear in $\mathcal{P}$ by \textsc{BalancePath}$(\mathcal{P}, \epsilon)$. 
Note that the order of shares in~$\mathcal{A}$ does not affect the fairness properties due to the identical valuation across all agents. 

\begin{algorithm}[t]
    \caption{\textsc{RecursiveBalance}$(\mathcal{A}, \epsilon)$.} \label{alg:recursivebalance}
    \textbf{Input}: Allocation $\mathcal{A} = (A_1, \ldots, A_n)$, $\epsilon \in (0, 1)$. \\
    \textbf{Output}: Allocation $\mathcal{A} = (A_1, \ldots, A_n)$. \\
    \vspace{-3.5mm}
    \begin{algorithmic}[1]
        \WHILE{$\frac{\mmax{\mathcal{A}}}{\mmin{\mathcal{A}}} > 2 + \epsilon$}
            \STATE $\mathcal{A} \leftarrow \textsc{Balance}(\mathcal{A}, \epsilon)$
        \ENDWHILE
        \STATE \textbf{return} $\mathcal{A}$
    \end{algorithmic}
\end{algorithm}

\begin{algorithm}[t]
    \caption{\textsc{Balance}$(\mathcal{A}, \epsilon)$.} \label{alg:balance}
    \textbf{Input}: Allocation $\mathcal{A} = (A_1, \ldots, A_n)$, $\epsilon \in (0, 1)$.\\
    \textbf{Output}: Allocation $\mathcal{A} = (A_1, \ldots, A_n)$.  \\
    \vspace{-3.5mm}
    \begin{algorithmic}[1]
        \STATE $\mathcal{P} \leftarrow \textsc{MinMaxPath}(\mathcal{A})$
        \STATE $\mathcal{A} \leftarrow (\mathcal{A} \setminus \mathcal{P}) \cup \textsc{BalancePath}(\mathcal{P}, \epsilon)$
        \STATE \textbf{return} $\mathcal{A}$
    \end{algorithmic}
\end{algorithm}

\begin{algorithm}[t]
    \caption{\textsc{BalancePath}$(\mathcal{P}, \epsilon)$.}    \label{alg:balancepath}
    \textbf{Input}: List of shares $\mathcal{P} = (P_1, \ldots, P_d)$, $\epsilon \in (0, 1)$. \\
    \textbf{Output}: List of $d$ shares $(P_1, \ldots, P_d)$. \\
    \textbf{Initialization}: $\gamma \leftarrow \val{P_d}$. \\
    \vspace{-3.5mm}
    \begin{algorithmic}[1]
    \FOR{$i = 1, \ldots, d - 1$}
        \IF{$\val{P_i} \geq \frac{\gamma}{2+\epsilon}$} \label{ln:case1}
            \STATE \textbf{return} $(P_1, \ldots, P_d)$ \label{ln:case1return}
        \ENDIF
        \STATE $P^* \leftarrow P_i \cup P_{i+1}$
        \IF{$\val{P^*} < \frac{2\gamma}{2+\epsilon}$} \label{ln:case2}
            \STATE $P_i \leftarrow P^*$
            \STATE $r \leftarrow$ any vertex in $P_d$
            \STATE $(P_{i+1}, P_d) \leftarrow \textsc{Divide}\left(P_d, N, \frac{\gamma}{3}, r\right)$
            \STATE \textbf{return} $(P_1, \ldots, P_d)$ \label{ln:case2return}
        \ENDIF
        \IF {$i = d - 1$}
            \STATE $r \leftarrow$ any vertex in $P_d$
        \ELSE
            \STATE $r \leftarrow$ any vertex in $P_{i+1} \cap P_{i+2}$
        \ENDIF
        \STATE $(P_i, P_{i+1}) \leftarrow \textsc{Divide}\left(P^*, N, \frac{\gamma}{2+\epsilon}, r\right)$ \label{ln:value}
    \ENDFOR
    \STATE \textbf{return} $(P_1, \ldots, P_d)$ \label{ln:case3return}
    \end{algorithmic}
\end{algorithm}

Now, \textsc{BalancePath} (\Cref{alg:balancepath}) does the bulk of the work. 
This algorithm adjusts the shares in $\mathcal{P} = (P_1, \ldots, P_d)$ so that their values meet certain criteria. 
Let $\gamma = \val{P_d}$ and $\widehat{P}_1 = P_1$.\footnote{We reuse $P_i$ for $P_i'$ and $\widehat{P}_i$ in the pseudocode for simplicity; however, we differentiate them in the main text for clarity.} 
For each $i$ from $1$ to $d - 1$, the algorithm does one of the following unless it is terminated prematurely via Case~1 or Case~2.

\begin{itemize}
    \item \textbf{Case 1: The value of $\widehat{P}_i$ is at least $\frac{\gamma}{2+\epsilon}$.} \\
    Set $P_i' = \widehat{P}_i$ and $P_j' = P_j$ for all $j \in \{i + 1,\dots, d\}$. Terminate the algorithm by returning $(P_1', \ldots, P_d')$.
    \item \textbf{Case 2: The value of $\widehat{P}_i$ is less than $\frac{\gamma}{2+\epsilon}$ and the value of $\widehat{P}_i \cup P_{i+1}$ is less than $\frac{2\gamma}{2+\epsilon}$.} \\
    Set $P_i' = \widehat{P}_i \cup P_{i+1}$. 
    Set $(P_{i+1}', P_d')$ to be the output of $\textsc{Divide}(P_d, N, \frac{\gamma}{3}, r)$, where $r$ is any point in~$P_d$; note that this call to \textsc{Divide} is valid because $P_d$ has value $\gamma$. 
    Set $P_j' = P_j$ for all $j \in \{i+2,\dots, d-1\}$. 
    Terminate the algorithm by returning $(P_1', \ldots, P_d')$.
    \item \textbf{Case 3: The value of $\widehat{P}_i$ is less than $\frac{\gamma}{2+\epsilon}$ and the value of $\widehat{P}_i \cup P_{i+1}$ is at least $\frac{2\gamma}{2+\epsilon}$.} \\
    Consider the graph $P^* = \widehat{P}_i \cup P_{i+1}$. Set $(P_i', \widehat{P}_{i+1})$ to be the output of $\textsc{Divide}(P^*, N, \frac{\gamma}{2+\epsilon}, r)$, where $r$ is any point belonging to both $P_{i+1}$ and $P_{i+2}$ (unless $i = d - 1$, in which case $r$ is any point in $P_d$); note that this call to \textsc{Divide} is valid because $P^*$ has value at least $\frac{2\gamma}{2+\epsilon}$.
    The choice of $r$ ensures that, if $i < d-1$, $\widehat{P}_{i+1}$ and $P_{i+2}$ share at least one point.
    Continue with the next~$i$ by incrementing it by $1$.
\end{itemize}
If the algorithm still has not terminated after $i = d - 1$, set $P_d' = \widehat{P}_d$ and return $(P_1', \ldots, P_d')$.

We claim that the algorithm \textsc{RecursiveBalance} terminates in polynomial time if it receives a $4$-EF allocation as input (provided by \Cref{thm:idenvalue_general-4ef}), and upon termination the algorithm returns a $(2 + \epsilon)$-EF allocation.

\begin{thm} \label{thm:idenvalue_general-2pluseps}
Given an instance of graphical cake cutting consisting of $n$ agents with identical valuations, there exists an algorithm that, for any $\epsilon > 0$, computes a $(2+\epsilon)$-EF allocation in time polynomial in $n$, $1/\epsilon$, and the size of $G$.
\end{thm}

To establish the proof of \Cref{thm:idenvalue_general-2pluseps}, we will work ``inside-out'': establish properties of \textsc{BalancePath} (\Cref{alg:balancepath}), \textsc{Balance} (\Cref{alg:balance}), and \textsc{RecursiveBalance} (\Cref{alg:recursivebalance}) in this order. 
Throughout the proofs, we will assume that the inputs of the algorithms are \emph{pseudo $4$-EF}---the definition is given below.

For a parameter $\alpha\ge 1$, we say that an allocation $\mathcal{A}$ is \emph{pseudo $\alpha$-EF} if $\mmin{\mathcal{A}} > 0$ and $\mmin{\mathcal{A}'} \geq \mmax{\mathcal{A}'} / \alpha$, where $\mathcal{A}'$ is defined to be an allocation after removing one share of the lowest value in $\mathcal{A}$ (if there is more than one such share, the definition is independent of which of those shares is removed). 
In other words, the concept of ``pseudo $\alpha$-EF'' ignores the effect of one share with the lowest value. 
We apply the analogous definition to a minimum-maximum path~$\mathcal{P}$. 
Note that an $\alpha$-EF allocation is also pseudo $\alpha$-EF. 
For any minimum-maximum path $\mathcal{P}$ of $\mathcal{A}$, $\mathcal{P}$ is pseudo $\alpha$-EF if $\mathcal{A}$ is, and $\mathcal{P}$ is $\alpha$-EF if and only if $\mathcal{A}$ is.

First, we establish properties satisfied by the output of \textsc{BalancePath}.

\begin{lem}\label{lem:balancepath}
Let $\epsilon \in (0, 1)$, and let $\mathcal{P} = (P_1, \ldots, P_d)$ be a minimum-maximum path of some allocation $\mathcal{A}$ that is pseudo $4$-EF but not $(2+\epsilon)$-EF. Let $\gamma = \mmax{\mathcal{P}}$.
Then, the following statements regarding $\mathcal{P}$ hold:
\begin{itemize}
    \item $0 < \val{P_1} < \frac{\gamma}{2+\epsilon}$,
    \item $\frac{\gamma}{4} \leq \val{P_j} \leq \gamma$ for all $j \in \{2,\dots, d-1\}$, and
    \item $\val{P_d} = \gamma$.
\end{itemize}
Moreover, if $\mathcal{P}' = (P'_1, \ldots, P'_d)$ is the output of $\textsc{BalancePath}(\mathcal{P}, \epsilon)$, then at least one of the following three cases holds:
\begin{itemize}
    \item Case 1: There exists $i \in \{2,\dots, d-1\}$ such that
    \begin{itemize}
        \item $\frac{\gamma}{2+\epsilon} \leq \val{P'_j} < \frac{2\gamma}{2+\epsilon}$ for all $j \in \{1,\dots, i-1\}$,
        \item $\frac{\gamma}{2+\epsilon} \leq \val{P'_i} < \val{P_i}$, and 
        \item $\val{P'_j} = \val{P_j}$ for all $j \in \{i+1,\dots, d\}$.
    \end{itemize}
    \item Case 2: There exists $i \in \{1,\dots, d-2\}$ such that
    \begin{itemize}
        \item $\frac{\gamma}{4} \leq \val{P'_j} < \frac{2\gamma}{2+\epsilon}$ for all $j \in \{1,\dots, i+1\} \cup \{d\}$, and
        \item $\val{P'_j} = \val{P_j}$ for all $j \in \{i+2,\dots, d-1\}$.
    \end{itemize}
    \item Case 3: 
    \begin{itemize}
        \item $\frac{\gamma}{2+\epsilon} \leq \val{P'_j} < \frac{2\gamma}{2+\epsilon}$ for all $j \in \{1,\dots, d-1\}$, and 
        \item $0 < \val{P'_d} < \val{P_d} = \gamma$.
    \end{itemize}
\end{itemize}
\end{lem}

\begin{proof}
First, we prove the statements regarding $\mathcal{P}$. 
If $\mathcal{A}$ (and hence $\mathcal{P}$) is not $(2+\epsilon)$-EF, then the smallest share, $P_1$, has value less than $\frac{\gamma}{2+\epsilon}$.
Since $\mathcal{A}$ (and hence $\mathcal{P}$) is pseudo $4$-EF, $P_1$ must have positive value and all other $P_j$'s must have value at least $\frac{\gamma}{4}$ and at most $\gamma$. 
Finally, $\val{P_d} = \mmax{\mathcal{P}} = \gamma$.

Next, we prove the results related to $\mathcal{P}'$. Cases 1, 2 and~3 correspond to the output of \textsc{BalancePath} in Lines \ref{ln:case1return}, \ref{ln:case2return}, and \ref{ln:case3return}, respectively.\footnote{Please note the notational differences between the main text and the pseudocode. We refer to the notations in the main text.} 
For each~$i\in \{1,\dots,d-1\}$, if the value of $\widehat{P}_i$ is less than $\frac{\gamma}{2+\epsilon}$ and the value of $\widehat{P}_i \cup P_{i+1}$ is at least $\frac{2\gamma}{2+\epsilon}$, then $\frac{\gamma}{2+\epsilon} \leq \val{P_i'} < \frac{2\gamma}{2+\epsilon}$ by $\textsc{Divide}(P^*, N, \frac{\gamma}{2+\epsilon}, r)$ and \Cref{lem:divide}. 
Hence,
\[
\val{\widehat{P}_{i+1}} = \val{\widehat{P}_i \cup P_{i+1}} - \val{P_i'} > 0,\]
and since $\val{\widehat{P}_i} < \frac{\gamma}{2+\epsilon} \leq \val{P_i'} $, we have
\begin{equation}
\label{eq:Piplusone}
\val{\widehat{P}_{i+1}} = \val{P_{i+1}} - (\val{P_i'} - \val{\widehat{P}_i}) < \val{P_{i+1}}.
\end{equation}
Putting the two bounds on $\val{\widehat{P}_{i+1}}$ together, we get $0 < \val{\widehat{P}_{i+1}} < \val{P_{i+1}} \leq \gamma$.

If there is some smallest $i \in \{2,\dots, d-1\}$ such that the value of $\widehat{P}_i$ is at least $\frac{\gamma}{2+\epsilon}$ (note that this is not possible for $i = 1$ since $\val{\widehat{P}_1} = \val{P_1} < \frac{\gamma}{2+\epsilon}$), then we have $\frac{\gamma}{2+\epsilon} \leq \val{P'_j} < \frac{2\gamma}{2+\epsilon}$ for all $j \in \{1,\dots, i-1\}$ by \Cref{lem:divide}, and $\frac{\gamma}{2+\epsilon} \leq \val{P'_i} < \val{P_i}$ by \eqref{eq:Piplusone} and the definition of $i$ (since the two early termination conditions were not triggered for $i-1$ and $P'_i = \widehat{P}_i$). 
The remaining shares are unchanged, i.e., $\val{P'_j} = \val{P_j}$ for all $j \in \{i+1,\dots, d\}$. 
This corresponds to Case~1.

Else, if there is some smallest $i \in \{1,\dots, d-2\}$ such that the value of $\widehat{P}_i \cup P_{i+1}$ is less than $\frac{2\gamma}{2+\epsilon}$ (note that this is not possible for $i = d-1$ since $\val{\widehat{P}_{d-1} \cup P_d} \geq \val{P_d} = \gamma > \frac{2\gamma}{2+\epsilon}$), then we have 
\[
\frac{\gamma}{4} \leq \frac{\gamma}{2+\epsilon} \leq \val{P'_j} < \frac{2\gamma}{2+\epsilon}
\]
for all $j \in \{1,\dots, i-1\}$ by \Cref{lem:divide}. 
Now,
$
\frac{\gamma}{4} \leq \val{P_{i+1}} \leq \val{\widehat{P}_i \cup P_{i+1}} < \frac{2\gamma}{2+\epsilon},
$
where the first inequality follows from the first paragraph of this proof.
Since $P'_i$ is set to $\widehat{P}_i \cup P_{i+1}$, we have 
\[
\frac{\gamma}{4} \leq \val{P'_i} < \frac{2\gamma}{2+\epsilon}.
\]
Also, since $(P'_{i+1}, P'_d) = \textsc{Divide}(P_d, N, \frac{\gamma}{3}, r)$, it holds that
\[
\frac{\gamma}{4} \leq \frac{\gamma}{3} \leq \val{P_j'} \leq \frac{2\gamma}{3} < \frac{2\gamma}{2+\epsilon}
\]
for $j \in \{i+1, d\}$ by \Cref{lem:divide} and the fact that $\val{P_d} = \gamma$. 
The remaining shares are unchanged, i.e., $\val{P'_j} = \val{P_j}$ for all $j \in \{i+2,\dots, d-1\}$. This corresponds to Case~2.

Finally, suppose that the value of $\widehat{P}_i$ is less than $\frac{\gamma}{2+\epsilon}$ and the value of $\widehat{P}_i \cup P_{i+1}$ is at least $\frac{2\gamma}{2+\epsilon}$ for all $i \in \{1,\dots, d-1\}$.
By \Cref{lem:divide}, $\frac{\gamma}{2+\epsilon} \leq \val{P'_j} < \frac{2\gamma}{2+\epsilon}$ for all $j \in \{1,\dots, d-1\}$.
Moreover, since 
\[
\frac{2\gamma}{2+\epsilon} \le \val{\widehat{P}_{d-1} \cup P_d} \le \val{\widehat{P}_{d-1}} + \val{P_d} < \frac{\gamma}{2+\epsilon} + \gamma
\]
and
$\frac{\gamma}{2+\epsilon} \leq \val{P'_{d-1}} < \frac{2\gamma}{2+\epsilon}$, and 
$\val{P'_d} = \val{\widehat{P}_{d-1} \cup P_d} - \val{P'_{d-1}}$, we have $0 < \val{P'_d} < \gamma = \val{P_d}$. 
This corresponds to Case~3.
\end{proof}

Having analyzed the output of \textsc{BalancePath}, we next establish properties satisfied by the output of \textsc{Balance}.
For an allocation $\mathcal{A}$, let 
$\mathcal{I}(\mathcal{A}) = \left[ \frac{\max (\mathcal{A})}{2+\epsilon}, \max (\mathcal{A}) \right]$ and $\mathcal{J}(\mathcal{A}) = \left[ \frac{2 \max (\mathcal{A})}{2+\epsilon}, \max (\mathcal{A}) \right]$.
Let $\mathcal{N}(\mathcal{A}, I)$ be the number of shares in $\mathcal{A}$ having values in the interval $I$, and let $\mathcal{N}(\mathcal{A}) = \mathcal{N}(\mathcal{A}, \mathcal{I}(\mathcal{A}))$.

\begin{lem}\label{lem:balance}
Let $\epsilon \in (0, 1)$, and let $\mathcal{A} = (A_1, \ldots, A_n)$ be an allocation that is pseudo $4$-EF but not $(2+\epsilon)$-EF. 
Let $\mathcal{A}' = (A'_1, \ldots, A'_n)$ be the output of $\textsc{Balance}(\mathcal{A}, \epsilon)$. 
Without loss of generality, assume that both $\mathcal{A}$ and $\mathcal{A}'$ have the shares arranged in ascending order of values. 
Then $\mathcal{A}'$ is pseudo $4$-EF, and at least one of the following two cases holds:
\begin{itemize}
    \item Case~(i): $\mathcal{N}(\mathcal{A}', \mathcal{I}(\mathcal{A})) > \mathcal{N}(\mathcal{A})$, and for all $j$ such that $\mu(A'_j) \in \mathcal{J}(\mathcal{A})$, we have $\mu(A'_j) \leq \mu(A_j)$.
    \item Case~(ii): For all $j$ such that $\mu(A'_j) \in \mathcal{J}(\mathcal{A})$, we have $\mu(A'_j) \leq \mu(A_{j-1})$.
\end{itemize}
\end{lem}

\begin{proof}
Since $\mathcal{A}$ is pseudo $4$-EF but not $(2+\epsilon)$-EF, we can apply \Cref{lem:balancepath} on its minimum-maximum path denoted by $\mathcal{P}$.
Note that the values of the shares in $\mathcal{A} \setminus \mathcal{P}$ remain unchanged, while all shares in $\mathcal{P}'$ have values in the range $[\frac{\gamma}{4}, \gamma]$, except possibly $P'_d$ which has value in the range $(0, \gamma]$. 
This means that $\mathcal{A}'$ is pseudo $4$-EF. 

We now show that at least one of the two cases in the lemma statement holds.
First, suppose that either $\mathcal{P}'$ falls under Case~1 of \Cref{lem:balancepath}, or $\mathcal{P}'$ falls under Case~3 with the additional condition that $\frac{\gamma}{2+\epsilon} \leq \val{P_d'} < \gamma$.
There exists $i \in \{2,\dots, d\}$ such that only $P_1, \ldots, P_i$ are changed to $P_1', \ldots, P_i'$, while the rest of the shares remain unchanged.
Let $\widehat{\mathcal{P}} = (P_1, \ldots, P_i)$ and $\widehat{\mathcal{P}}' = (P_1', \ldots, P_i')$. 
We have $\val{P_1} \not\in \mathcal{I}(\mathcal{A})$ because $\mathcal{A}$ is not $(2+\epsilon)$-EF, so $\mathcal{N}(\widehat{\mathcal{P}}, \mathcal{I}(\mathcal{A})) \leq i-1$, while $\val{P_j'} \in \mathcal{I}(\mathcal{A})$ for all $j \in \{1,\dots, i\}$, so $\mathcal{N}(\widehat{\mathcal{P}}', \mathcal{I}(\mathcal{A})) = i$. 
As the rest of the shares in $\mathcal{P}$ and $\mathcal{A}$ remain unchanged, this implies $\mathcal{N}(\mathcal{A}', \mathcal{I}(\mathcal{A})) > \mathcal{N}(\mathcal{A})$. 
Furthermore, the only share in $\widehat{\mathcal{P}}'$ with a value that is potentially in $\mathcal{J}(\mathcal{A})$ is $P_i'$, but since $\val{P_i'} < \val{P_i}$, we have $\mu(A'_j) \leq \mu(A_j)$ for all $j$ such that $\mu(A'_j) \in \mathcal{J}(\mathcal{A})$. 
This corresponds to Case~(i).

Next, suppose that either $\mathcal{P}'$ falls under Case~2, or $\mathcal{P}'$ falls under Case~3 with the additional condition that $0 < \val{P_d'} < \frac{\gamma}{2+\epsilon}$.
There exists $i \in \{1,\dots, d-2\}$ such that only $P_1, \ldots, P_{i+1}$ and $P_d$ are changed to $P_1', \ldots, P_{i+1}'$ and $P_d'$.  
Then we have $\val{P_j'} \not\in \mathcal{J}(\mathcal{A})$ for all $j \in \{1,\dots, i+1\} \cup \{d\}$.
Hence, for $j$ such that $\mu(P'_j) \in \mathcal{J}(\mathcal{A})$, it holds that $j\in\{i+2,\dots,d-1\}$ and $\mu(P_j) = \mu(P'_j) \in \mathcal{J}(\mathcal{A})$.
Moreover, $\val{P_d} = \mmax{\mathcal{A}} \in \mathcal{J}(\mathcal{A})$.
Since both $\mathcal{A}$ and $\mathcal{A}'$ have the shares arranged in ascending order of values and the values of the shares in $\mathcal{A}\setminus\mathcal{P}$ remain unchanged, it follows that $\mu(A'_j) \leq \mu(A_{j-1})$ for all $j$ such that $\mu(A'_j) \in \mathcal{J}(\mathcal{A})$.
This corresponds to Case~(ii).
\end{proof}

By leveraging \Cref{lem:balance}, we can bound the number of calls to \textsc{Balance} in \textsc{RecursiveBalance}.

\begin{lem} \label{lem:recursivebalance}
Given an instance of graphical cake cutting consisting of $n$ agents with identical valuations, a $4$-EF allocation $\mathcal{A}$ of the cake, and any $\epsilon \in (0, 1)$, the algorithm \textsc{RecursiveBalance} terminates after at most $\frac{5n^2}{\epsilon}$ calls to $\textsc{Balance}$.
\end{lem}

\begin{proof}
Let $\mathcal{A}^0 = \mathcal{A}$, and $\mathcal{A}^{t+1} = \textsc{Balance}(\mathcal{A}^t, \epsilon)$ for all $t \geq 0$. 
If there is some $t$ such that $\mathcal{A}^t$ is $(2+\epsilon)$-EF, then the algorithm \textsc{RecursiveBalance} terminates. 
By \Cref{lem:balance}, every $\mathcal{A}^t$ is pseudo $4$-EF, and if $\mathcal{A}^t$ is not $(2+\epsilon)$-EF, then either Case~(i) or Case~(ii) holds when $\mathcal{A}^t$ is used as an input to \textsc{Balance}, with output $\mathcal{A}^{t+1}$.
As a consequence of \Cref{lem:balance}, we have $\mmax{\mathcal{A}^{t+1}} \leq \mmax{\mathcal{A}^t}$.
For an allocation $\mathcal{B}$, let us say that $\mathcal{B}$ \emph{follows Case~(i)} (resp., \emph{Case~(ii)}) if the pair $(\mathcal{B}, \mathcal{B}')$ fulfills the conditions of Case~(i) (resp., Case~(ii)), where $\mathcal{B}' = \textsc{Balance}(\mathcal{B},\epsilon)$.

We claim that \textsc{RecursiveBalance} terminates if there are $n$ consecutive allocations following Case~(i). 
Consider some $\mathcal{A}^t$ and assume that for every $k \in \{0,\dots, n-1\}$, $\mathcal{A}^{t+k}$ follows Case~(i) and is not $(2+\epsilon)$-EF.
For each $k$, we have 
\[
\mathcal{N}(\mathcal{A}^{t+k+1}, \mathcal{I}(\mathcal{A}^{t+k})) > \mathcal{N}(\mathcal{A}^{t+k})\] 
by \Cref{lem:balance}, and \[
\mathcal{N}(\mathcal{A}^{t+k+1}) \geq \mathcal{N}(\mathcal{A}^{t+k+1}, \mathcal{I}(\mathcal{A}^{t+k}))
\]
holds due to the fact that $\mmax{\mathcal{A}^{t+k+1}} \leq \mmax{\mathcal{A}^{t+k}}$. 
This gives the chain 
\[
\mathcal{N}(\mathcal{A}^t) < \mathcal{N}(\mathcal{A}^{t+1}) < \cdots < \mathcal{N}(\mathcal{A}^{t+n}). 
\]
Since each $\mathcal{N}(\mathcal{A}^{t+k})$ is an integer bounded by $0$ and $n$, the chain forces $\mathcal{N}(\mathcal{A}^{t+k}) = k$ for all $k \in \{0,\dots, n\}$.
In particular, $\mathcal{N}(\mathcal{A}^{t+n}) = n$, which implies that $\mathcal{A}^{t+n}$ is $(2+\epsilon)$-EF, and the algorithm terminates.

Let $\mathcal{A}^t$ and $k > 0$ be given.
For each $i \in \{0,\dots, k\}$, consider the number of shares the allocation $\mathcal{A}^{t+i}$ has in the interval $\mathcal{J}(\mathcal{A}^t)$, i.e., consider $z_i = \mathcal{N}(\mathcal{A}^{t+i}, \mathcal{J}(\mathcal{A}^t))$. 
By \Cref{lem:balance}, if $\mathcal{A}^{t+i}$ follows Case~(i), then for all $j$ such that $\val{A_j^{t+i+1}}\in \mathcal{J}(\mathcal{A}^{t+i})$, we have $\val{A_j^{t+i+1}} \le \val{A_j^{t+i}}$.
Note that if $\val{A_j^{t+i+1}}\in \mathcal{J}(\mathcal{A}^{t})$, then it must also be that $\val{A_j^{t+i+1}}\in \mathcal{J}(\mathcal{A}^{t+i})$, because $\mmax{\mathcal{A}^{t}} \geq \mmax{\mathcal{A}^{t+i}} \geq \mmax{\mathcal{A}^{t+i+1}}$.
Hence, if $\mathcal{A}^{t+i}$ follows Case~(i), then for all~$j$ such that $\val{A_j^{t+i+1}}\in \mathcal{J}(\mathcal{A}^{t})$, we have $\val{A_j^{t+i+1}} \le \val{A_j^{t+i}}$.
This implies that $z_{i+1}\le z_i$ in this case.
Now, consider the case where $z_i > 0$ and $\mathcal{A}^{t+i}$ follows Case~(ii).
This means that for all $j$ such that $\mu(A_j^{t+i+1})\in\mathcal{J}(\mathcal{A}^{t+i})$, we have $\mu(A_j^{t+i+1}) \le \mu(A_{j-1}^{t+i})$.
Note again that if $\val{A_j^{t+i+1}}\in \mathcal{J}(\mathcal{A}^{t})$, then it must also be that $\val{A_j^{t+i+1}}\in \mathcal{J}(\mathcal{A}^{t+i})$, because $\mmax{\mathcal{A}^{t}} \geq \mmax{\mathcal{A}^{t+i}} \geq \mmax{\mathcal{A}^{t+i+1}}$.
Therefore, for all $j$ such that $\mu(A_j^{t+i+1})\in\mathcal{J}(\mathcal{A}^{t})$, we have $\mu(A_j^{t+i+1}) \le \mu(A_{j-1}^{t+i})$.
Combined with the assumption that $z_i > 0$, it follows that $z_{i+1} < z_i$ in this case.
Hence, if at least $n$ allocations among $\mathcal{A}^t, \mathcal{A}^{t+1}, \dots, \mathcal{A}^{t+k-1}$ follow Case~(ii), then $z_k = 0$ and consequently, $\mmax{\mathcal{A}^{t+k}} \leq \frac{2}{2+\epsilon}\mmax{\mathcal{A}^t}$.
Similarly, if at least $qn$ allocations among $\mathcal{A}^t, \mathcal{A}^{t+1}, \dots, \mathcal{A}^{t+k-1}$ follow Case~(ii) for some positive integer~$q$, then $\mmax{\mathcal{A}^{t+k}} \leq \left( \frac{2}{2+\epsilon} \right)^q \mmax{\mathcal{A}^t}$.

We are now ready to show that \textsc{RecursiveBalance} terminates after at most $qn^2$ calls to \textsc{Balance}, where $q = \left\lfloor 5/\epsilon \right\rfloor$. Suppose on the contrary that the algorithm still has not terminated at $\mathcal{A}^{qn^2}$.
If fewer than $qn$ allocations among $\mathcal{A}^0, \mathcal{A}^1, \dots, \mathcal{A}^{qn^2-1}$ follow Case~(ii), then by a counting argument, there must be $n$ consecutive allocations in this sequence that follow Case~(i), and hence the algorithm would have terminated. 
Therefore, at least $qn$ allocations among $\mathcal{A}^0, \mathcal{A}^1, \dots, \mathcal{A}^{qn^2-1}$ follow Case~(ii), so by the previous paragraph, $\mmax{\mathcal{A}^{qn^2}} \leq \left( \frac{2}{2+\epsilon} \right)^q \mmax{\mathcal{A}^0}$.
Since $\mathcal{A}^0$ is $4$-EF, we must have $\mmax{\mathcal{A}^0} \leq \frac{4}{n+3}$ (otherwise, every share has value more than $\frac{1}{n+3}$ and the total value of all $n$ shares is more than $\frac{4}{n+3} + (n-1)\cdot \frac{1}{n+3} = 1$, a contradiction). 
Now, the function $f(x) = (1+x)^{1/x}$ is decreasing in $(0,\infty)$---to see this, note that $\ln f(x) = \frac{1}{x}\ln(1+x) = \frac{\ln(1+x)-\ln 1}{x}$, which is the slope of the line through $(1,\ln 1)$ and $(1+x,\ln(1+x))$, and must be decreasing because the $\ln$ function is concave.
Hence, $( 1 + \frac{\epsilon}{2})^{2/\epsilon} = f(\epsilon/2) \ge f(1) = 2$, and so
\[
\left( \frac{2+\epsilon}{2} \right)^q = \left( 1 + \frac{\epsilon}{2} \right)^q \geq \left( 1 + \frac{\epsilon}{2} \right)^{4 / \epsilon} \geq 2^2 = 4,
\]
which implies that $\left(\frac{2}{2+\epsilon}\right)^q \leq \frac{1}{4}$.
It follows that 
\[
\mmax{\mathcal{A}^{qn^2}} \leq \left(\frac{1}{4}\right)\left(\frac{4}{n+3}\right) = \frac{1}{n+3} \leq \frac{2}{2n-1}.
\]
However, this means that $\mathcal{A}^{qn^2}$ is $2$-EF (otherwise, some share has value less than $\frac{1}{2n-1}$ and the total value of all $n$ shares is less than $\frac{1}{2n-1} + (n-1)\cdot\frac{2}{2n-1} = 1$, a contradiction) and hence $(2+\epsilon)$-EF. 
This is a contradiction since the algorithm should have terminated at $\mathcal{A}^{qn^2}$.
\end{proof}

We now use \Cref{lem:recursivebalance} to establish \Cref{thm:idenvalue_general-2pluseps}.

\begin{proof}[Proof of \Cref{thm:idenvalue_general-2pluseps}]
Without loss of generality, we may assume that $\epsilon \in (0, 1)$. 
Apply \textsc{IterativeDivide} to obtain a $4$-EF allocation (\Cref{thm:idenvalue_general-4ef}), then apply \textsc{RecursiveBalance} on this allocation to obtain a $(2 + \epsilon)$-EF allocation.
\textsc{IterativeDivide} and each iteration of \textsc{Balance} run in time polynomial in $n$, $1/\epsilon$, and the size of $G$, and the number of iterations of \textsc{Balance} is polynomial in $n$ and $1 / \epsilon$ by \Cref{lem:recursivebalance}.
The claimed overall running time follows.
\end{proof}

\subsection{Star Graphs}
Although \textsc{RecursiveBalance} provides a nearly $2$-EF allocation, the algorithm is rather complex and involves many steps.
Moreover, it seems difficult to improve the guarantee to \emph{exactly} $2$-EF via the algorithm, as the balancing of shares may be too slow and we may not reach a $2$-EF allocation in finite time.\footnote{For example, agents with shares of $1$~unit and $2 + \epsilon$ units respectively may repeatedly exchange and adjust their shares with each other (through multiple calls to \Cref{alg:balancepath}) but may never reach a multiplicative envy factor of $2$.}
We show next that, for the case of star graphs, it is possible to obtain a $2$-EF allocation using a simpler ``bag-filling'' approach.

Let $G$ be a star graph centered at vertex $v$. While there exists an edge $[u, v]$ in~$G$ with value at least $1/n$, find a point $x \in [u, v]$ such that $[u, x]$ has value $1/n$, and allocate $[u, x]$ to one of the agents. 
Remove this agent from consideration and remove the allocated segment from $G$; note that $G$ remains a star graph. 
Repeat this process on the remaining graph and agents until every edge has value less than $1/n$.
See \Cref{fig:stub} for an illustration.

At this point, we are left with a star of ``stubs''---edges with value less than $1/n$ each---and $k \in \{0, \dots, n\}$ agents who are not allocated any share yet.
If $k = 0$ we are done,\footnote{There may be a star of stubs of value $0$ left; in that case we append it to the last agent's share so that connectivity is retained.} so assume that $k\ge 1$. 
The total value of all stubs is exactly $k/n$, so there are more than $k$ stubs.
Make each stub a separate ``group'', then repeatedly merge two groups of the lowest value until there are exactly $k$ groups. 
Assign these $k$ groups to the $k$ agents; note that each group is connected by the vertex $v$. We claim that the resulting allocation is $2$-EF.

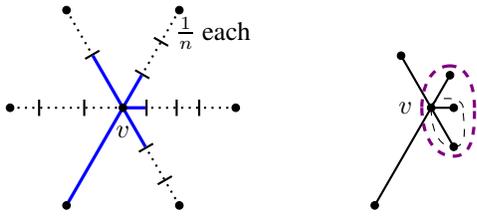
\begin{figure}[t]
\centering
  \begin{subfigure}[t]{.32\textwidth}
    \centering 
    \begin{tikzpicture}
    \usetikzlibrary{positioning}
    \node[label=below:$v$] at (0,0) (v) {};
    \node at (180:1.5cm) (v1) {};
    \node at (120:1.5cm) (v2) {};
    \node at (60:1.5cm) (v3) {};
    \node at (0:1.5cm) (v4) {};
    \node at (-60:1.5cm) (v5) {};
    \node at (-120:1.5cm) (v6) {};
    \node at (180:1.1cm) (x11) {};
    \node at (180:0.5cm) (x12) {};
    \node at (120:0.8cm) (x2) {};
    \node at (60:1.0cm) (x31) {};
    \node at (60:0.5cm) (x32) {};
    \node at (0:1.0cm) (x41) {};
    \node at (0:0.7cm) (x42) {};
    \node at (0:0.3cm) (x43) {};
    \node at (-60:1.1cm) (x51) {};
    \node at (-60:0.6cm) (x52) {};
    \draw [-|, dotted, thick]
        (v1.center) -- (x11.center);
    \draw [-|, dotted, thick]
        (x11.center) -- (x12.center);
    \draw [-, dotted, thick]
        (x12.center) -- (v.center);
    \draw [-|, dotted, thick]
        (v2.center) -- (x2.center);
    \draw [-|, dotted, thick]
        (v3.center) -- (x31.center)  node[pos=0.7, right]{$\frac{1}{n}$ each};
    \draw [-|, dotted, thick]
        (x31.center) -- (x32.center);
    \draw [-|, dotted, thick]
        (v4.center) -- (x41.center);
    \draw [-|, dotted, thick]
        (x41.center) -- (x42.center);
    \draw [-|, dotted, thick]
        (x42.center) -- (x43.center);
    \draw [-|, dotted, thick]
        (v5.center) -- (x51.center);
    \draw [-|, dotted, thick]
        (x51.center) -- (x52.center);
    \draw [very thick, blue]
        (v.center) -- (x2.center)
        (v.center) -- (x32.center)
        (v.center) -- (x43.center)
        (v.center) -- (x52.center)
        (v.center) -- (v6.center);
    \draw [black, fill=black]
        (v) circle [radius=0.05]
        (v1) circle [radius=0.05]
        (v2) circle [radius=0.05]
        (v3) circle [radius=0.05]
        (v4) circle [radius=0.05]
        (v5) circle [radius=0.05]
        (v6) circle [radius=0.05];
    \end{tikzpicture}
  \end{subfigure}
  \begin{subfigure}[t]{.32\textwidth}
    \centering 
    \begin{tikzpicture}
    \usetikzlibrary{positioning}
    \node[label=left:$v$] at (0,0) (v) {};
    \node at (180:1.5cm) (v1) {};
    \node at (120:1.5cm) (v2) {};
    \node at (60:1.5cm) (v3) {};
    \node at (0:1.5cm) (v4) {};
    \node at (-60:1.5cm) (v5) {};
    \node at (-120:1.5cm) (v6) {};
    \node at (180:1.0cm) (x11) {};
    \node at (180:0.6cm) (x12) {};
    \node at (120:0.8cm) (x2) {};
    \node at (60:1.0cm) (x31) {};
    \node at (60:0.5cm) (x32) {};
    \node at (0:1.1cm) (x41) {};
    \node at (0:0.7cm) (x42) {};
    \node at (0:0.3cm) (x43) {};
    \node at (-60:1.1cm) (x51) {};
    \node at (-60:0.6cm) (x52) {};
    \draw [black, fill=black]
        (v) circle [radius=0.05]
        (x2) circle [radius=0.05]
        (x32) circle [radius=0.05]
        (x43) circle [radius=0.05]
        (x52) circle [radius=0.05]
        (v6) circle [radius=0.05];
    \draw [thick]
        (v.center) -- (x2.center)
        (v.center) -- (x32.center)
        (v.center) -- (x43.center)
        (v.center) -- (x52.center)
        (v.center) -- (v6.center);
    \draw [dashed] 
        (v.center) to [out=-90,in=180] (x52.center)
        (x52.center) to [out=0,in=-90] (x43.east)
        (x43.east) to [out=90,in=90] (v.center);  
    \draw [dashed, very thick, violet] 
        (v.west) to [out=-90,in=180] (x52.south)
        (x52.south) to [out=0,in=-90] (x42.west)
        (x42.west) to [out=90,in=0] (x32.north)
        (x32.north) to [out=180,in=90] (v.west);
    \end{tikzpicture}
  \end{subfigure}
\caption{(Left) Segments of value $1/n$ each (dotted lines) are allocated to the agents. The remaining stubs (solid lines) have value less than $1/n$ each. (Right) Two groups of the lowest value are merged together repeatedly. The figure shows three final groups: one group of three stubs on the right and two groups of one stub each on the left.} \label{fig:stub}
\end{figure}

\begin{thm} \label{thm:idenvalue_star}
Given an instance of graphical cake cutting consisting of a star graph with $m$ edges and $n$ agents with identical valuations, there exists an algorithm that computes a $2$-EF allocation in time polynomial in $n$ and $m$.
\end{thm}

\begin{proof}
The running time claim is clear.

If no agent remains after assigning segments of value $1/n$ (i.e., $k=0$), then the allocation is envy-free and hence $2$-EF.

Assume that $k\ge 1$.
The total value of the star of stubs is $k/n$, so we cannot end the process with $k$ groups of value less than $1/n$ each.
Therefore, at some point, a newly-merged group has value at least $1/n$. 
Two groups of the lowest value are always the ones selected for the merge, so the value of each newly-merged group is at least that of any previously-merged group.
Let $g^*$ be the final group formed by merging two (now extinct) groups $g_1$ and $g_2$, and assume without loss of generality that $\val{g_1} \leq \val{g_2}$.
Not only does $g^*$ have the maximum value across all \emph{groups}, but the fact that $\val{g^*} \geq 1/n$ implies that it also has the maximum value across all \emph{shares} (including the $1/n$-value segments from the initial process).
Furthermore, it follows that $\val{g^*} \leq 2/n$; otherwise, we would have $\val{g_2} > 1/n$, in which case each of the $k$ groups would have value greater than $1/n$ and the sum of their values could not be exactly $k/n$.

We shall prove that the allocation is $2$-EF. 
Since $g^*$ has the maximum value across all shares, it suffices to prove that every agent receives a share worth at least $\val{g^*}/2$.
Let agent $i$ be given. 
If agent $i$ receives some group $g$, then $g$ must have value at least $\val{g_2}$ (otherwise $g$ would have been chosen for the final merge with $g_1$ instead of $g_2$).
This means that $\val{g} \geq (\val{g_1} + \val{g_2})/2 = \val{g^*}/2$. 
Else, agent $i$ receives a $1/n$-value segment; then the segment has value at least $\val{g^*}/2$ since $\val{g^*} \leq 2/n$. 
It follows that the allocation is $2$-EF, as desired.
\end{proof}

\section{Beyond One Connected Piece} \label{sec:beyond}

In previous sections, we made the crucial assumption that each agent necessarily receives a connected piece of the graphical cake.
We now relax this assumption and explore improvements in envy-freeness guarantees if each agent is allowed to receive a small number of connected pieces. 
As mentioned in \Cref{subsec:related}, extensive research has been done on (approximate) envy-free connected allocations of an \emph{interval cake}; it will be interesting to draw connections between these results and graphical cake cutting.

As an example, consider a star graph with edges $e_k$ for $k \in \{1, \ldots, m\}$. 
Rearrange the edges to form a path graph with the edges $e_1$ to $e_m$ going from left to right; the farthest end of each edge from the center vertex in the star graph is oriented towards the \emph{right} on the path graph.
Note that any segment along the path graph corresponds to at most two connected pieces in the star graph; see \Cref{fig:psn} for an illustration.\footnote{A segment can start and end at points within an edge and span across different edges---it need not start and end at vertices.}
It follows that if a connected $\alpha$-EF (resp., $\alpha$-additive-EF) allocation of an interval cake can be found in polynomial time, then there also exists a polynomial-time algorithm that finds an $\alpha$-EF (resp., $\alpha$-additive-EF) allocation of the star cake where each agent receives at most \emph{two} connected pieces. 

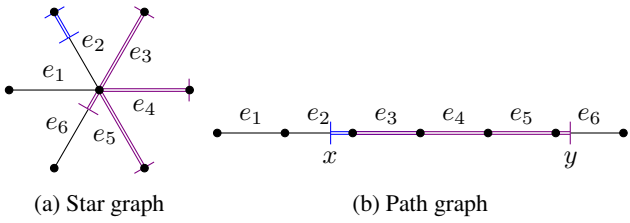
\begin{figure}[t]
\centering
  \begin{subfigure}[b]{.3\textwidth}
    \centering 
    \begin{tikzpicture}
    \usetikzlibrary{positioning}
    \node at (0,0) (v) {};
    \node at (180:1.2cm) (v1) {};
    \node at (120:1.2cm) (v2) {};
    \node at (60:1.2cm) (v3) {};
    \node at (0:1.2cm) (v4) {};
    \node at (-60:1.2cm) (v5) {};
    \node at (-120:1.2cm) (v6) {}; 
    \node at (120:0.8cm) (x2) {};
    \node at (-120:0.3cm) (x6) {};
    \draw [white]
        (v.center) -- (v1.center)  node[pos=0.5, above, black]{$e_1$} 
        (v.center) -- (v2.center)  node[pos=0.6, right, black]{$e_2$} 
        (v.center) -- (v3.center)  node[pos=0.45, right, black]{$e_3$} 
        (v.center) -- (v4.center)  node[pos=0.5, below, black]{$e_4$}
        (v.center) -- (v5.center)  node[pos=0.6, left, black]{$e_5$}
        (v.center) -- (v6.center)  node[pos=0.45, left, black]{$e_6$};
    \draw []
        (v.center) -- (v1.center)
        (v.center) -- (x2.center)
        (x6.center) -- (v6.center);
    \draw [|-|, double, blue]
        (x2.center) -- (v2.center);
    \draw [-|, double, violet]
        (v.center) -- (v3.center);
    \draw [-|, double, violet]
        (v.center) -- (v4.center);
    \draw [-|, double, violet]
        (v.center) -- (v5.center);
    \draw [-|, double, violet]
        (v.center) -- (x6.center);
    \draw [black, fill=black]
        (v) circle [radius=0.05]
        (v1) circle [radius=0.05]
        (v2) circle [radius=0.05]
        (v3) circle [radius=0.05]
        (v4) circle [radius=0.05]
        (v5) circle [radius=0.05]
        (v6) circle [radius=0.05];
    \end{tikzpicture}
    \caption{Star graph}
  \end{subfigure}
  \begin{subfigure}[b]{.5\textwidth}
    \centering 
    \begin{tikzpicture}
    \usetikzlibrary{positioning}
    \node at (180:2.7cm) (v0) {};
    \node at (180:1.8cm) (v1) {};
    \node at (180:0.9cm) (v2) {};
    \node at (0,0) (v3) {};
    \node at (0:0.9cm) (v4) {};
    \node at (0:1.8cm) (v5) {};
    \node at (0:2.7cm) (v6) {};
    \node[label=below:$x$] at (180:1.2cm) (x2) {};
    \node[label=below:$y$] at (0:2.0cm) (x6) {};
    \draw [white]
        (v0.center) -- (v1.center)  node[pos=0.5, above, black]{$e_1$}
        (v1.center) -- (v2.center)  node[pos=0.5, above, black]{$e_2$}
        (v2.center) -- (v3.center)  node[pos=0.5, above, black]{$e_3$}
        (v3.center) -- (v4.center)  node[pos=0.5, above, black]{$e_4$}
        (v4.center) -- (v5.center)  node[pos=0.5, above, black]{$e_5$}
        (v5.center) -- (v6.center)  node[pos=0.5, above, black]{$e_6$};
    \draw []
        (v0.center) -- (x2.center)
        (x6.center) -- (v6.center);
    \draw [|-, double, blue]
        (x2.center) -- (v3.center);
    \draw [-, double, violet]
        (v2.center) -- (v5.center);
    \draw [-|, double, violet]
        (v5.center) -- (x6.center);
    \draw [black, fill=black]
        (v0) circle [radius=0.05]
        (v1) circle [radius=0.05]
        (v2) circle [radius=0.05]
        (v3) circle [radius=0.05]
        (v4) circle [radius=0.05]
        (v5) circle [radius=0.05]
        (v6) circle [radius=0.05];
    \end{tikzpicture}
    \caption{Path graph}
  \end{subfigure}
\caption{(a) A star graph with six edges and (b) its corresponding path graph. A segment $[x, y]$ along the path graph (double lines) corresponds to at most two connected pieces in the star graph.} \label{fig:psn}
\end{figure}

We now consider an arbitrary graph $G$. 
Let $\phi$ be a bijection of the edges of $G$ onto the edges of a path graph with the same number of edges as $G$, where the orientation of the edges along the path can be chosen.
By an abuse of notation, let $\phi(G)$ be the corresponding path graph.
Define the \emph{path similarity number of $G$ associated with $\phi$}, denoted by $\psn{G, \phi}$, as the smallest number $k$ such that any segment along $\phi(G)$ corresponds to at most $k$ connected pieces in $G$.
Any results pertaining to connected allocations of an interval cake\footnote{For instance, \citet{BarmanKu22} presented a polynomial-time algorithm that computes a (roughly) $1/4$-additive-EF and $2$-EF allocation of an interval cake.} can be directly applied to allocations of $G$ in which each agent receives at most $\psn{G, \phi}$ connected pieces.

\begin{prop}
\label{prop:bijection}
Let $G$ be a graph and $\phi$ be a bijection of the edges of $G$ onto $\phi(G)$ with orientation such that $\phi$ can be found in polynomial time. If there exists a polynomial-time algorithm that computes a connected $\alpha$-EF allocation of an interval cake, then there exists a polynomial-time algorithm that computes an $\alpha$-EF allocation of a graphical cake $G$ where each agent receives at most $\psn{G, \phi}$ connected pieces. An analogous statement holds for $\alpha$-additive-EF.
\end{prop}

To complement \Cref{prop:bijection}, we provide upper bounds of $\psn{G, \phi}$ where $\phi$ can be computed in polynomial time.

\begin{thm} \label{thm:tree}
Let $G$ be a tree of height $h$. Then there exists a bijection $\phi$ that can be computed in polynomial time such that $\psn{G, \phi} \leq h + 1$.
\end{thm}

\begin{proof}
Define $\phi(G)$ as follows: arrange the edges from left to right according to their appearance in a depth-first search of $G$, and orient their directions so that the farthest point of each edge from the root vertex in $G$ appears towards the \emph{right} of the edge in $\phi(G)$. Clearly, $\phi$ can be computed in polynomial time. We claim that any segment along $\phi(G)$ corresponds to at most $h + 1$ connected pieces in $G$.

Let $[x, y]$ be a segment along $\phi(G)$, and consider its corresponding piece(s) in $G$---see \Cref{fig:psn_thm52} for an illustration. If $[x, y]$ corresponds to a path in $G$, then this path is a connected piece in $G$. Otherwise, as we traverse $[x, y]$ from left to right in $\phi(G)$, the corresponding traversal in $G$ will eventually return to some ancestor vertex before another branch in $G$ is searched; this process may be repeated. Therefore, there exists some point in $[x, y]$ that corresponds to a vertex $v$ in $G$ such that $[x, y]$ corresponds to a subgraph of the subtree rooted at $v$ (for example, in \Cref{fig:psn_thm52}, $v$ is the root of the tree $G$). Let $z_1 \in [x, y]$ be the leftmost point that corresponds to vertex~$v$ (in \Cref{fig:psn_thm52}, only one point in $[x,y]$ corresponds to $v$). Note that $[z_1, y]$ corresponds to a connected piece in $G$---this is because $z_1$ corresponds to the root vertex of the subtree and a depth-first search from $z_1$ to $y$ ensures that the subgraph is connected. 

By repeating this process on $[x, z_i]$, we can find a sequence of points $z_1, \ldots, z_k$ in $[x, y]$ such that $[x, z_k], [z_k, z_{k-1}], \ldots, [z_2, z_1], [z_1, y]$ each corresponds to a connected piece in $G$---note that there are at most $k + 1$ connected pieces. Furthermore, $z_1, \ldots, z_k$ correspond to a chain of vertices with an ancestor-descendant relationship, i.e., $z_i$ is an ancestor of $z_j$ for $i < j$ in the corresponding graph $G$. Since none of the $z_i$'s correspond to leaf vertices, the chain has at most $h$ vertices, i.e., $k \leq h$. Therefore, there are at most $h + 1$ connected pieces in $G$.
\end{proof}

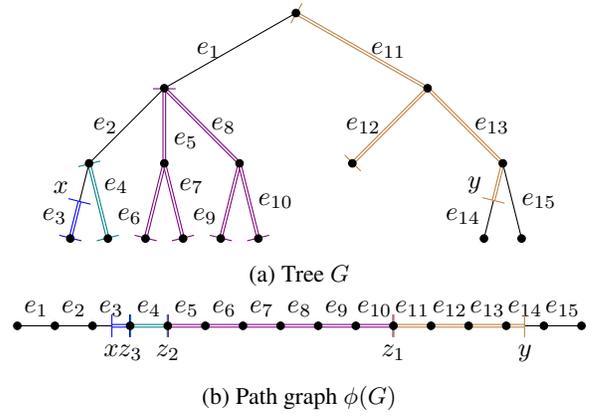
\begin{figure}[t]
\centering
  \begin{subfigure}[b]{.45\textwidth}
    \centering 
    \begin{tikzpicture}[scale=0.9]
    \usetikzlibrary{positioning}
    \node at (0,0) (v0) {};
    \node at (-1.75cm,-1cm) (v1) {};
    \node at (-2.75cm,-2cm) (v2) {};
    \node at (-3.0cm,-3cm) (v3) {};
    \node at (-2.5cm,-3cm) (v4) {};
    \node at (-1.75cm,-2cm) (v5) {};
    \node at (-2.0cm,-3cm) (v6) {};
    \node at (-1.5cm,-3cm) (v7) {};
    \node at (-0.75cm,-2cm) (v8) {};
    \node at (-1.0cm,-3cm) (v9) {};
    \node at (-0.5cm,-3cm) (v10) {};
    \node at (1.75cm,-1cm) (v11) {};
    \node at (0.75cm,-2cm) (v12) {};
    \node at (2.75cm,-2cm) (v13) {};
    \node at (2.5cm,-3cm) (v14) {};
    \node at (3.0cm,-3cm) (v15) {};
    \node[label={[xshift=-0.25cm,yshift=-0.15cm]\footnotesize $x$}] at (-2.875cm,-2.5cm) (x) {};
    \node[label={[xshift=-0.25cm,yshift=-0.15cm]\footnotesize $y$}] at (2.625cm,-2.5cm) (y) {};
    \draw [white]
        (v0.center) -- (v1.center)  node[pos=0.5, left, black]{\footnotesize $e_1$}
        (v1.center) -- (v2.center)  node[pos=0.5, left, black]{\footnotesize $e_2$}
        (v2.center) -- (v3.center)  node[pos=0.7, left, black]{\footnotesize $e_3$}
        (v2.center) -- (v4.center)  node[pos=0.3, right, black]{\footnotesize $e_4$}
        (v1.center) -- (v5.center)  node[pos=0.7, right, black]{\footnotesize $e_5$}
        (v5.center) -- (v6.center)  node[pos=0.7, left, black]{\footnotesize $e_6$}
        (v5.center) -- (v7.center)  node[pos=0.3, right, black]{\footnotesize $e_7$}
        (v1.center) -- (v8.center)  node[pos=0.5, right, black]{\footnotesize $e_8$}
        (v8.center) -- (v9.center)  node[pos=0.7, left, black]{\footnotesize $e_9$}
        (v8.center) -- (v10.center)  node[pos=0.5, right, black]{\footnotesize $e_{10}$}
        (v0.center) -- (v11.center)  node[pos=0.5, right, black]{\footnotesize $e_{11}$}
        (v11.center) -- (v12.center)  node[pos=0.5, left, black]{\footnotesize $e_{12}$}
        (v11.center) -- (v13.center)  node[pos=0.5, right, black]{\footnotesize $e_{13}$}
        (v13.center) -- (v14.center)  node[pos=0.7, left, black]{\footnotesize $e_{14}$}
        (v13.center) -- (v15.center)  node[pos=0.5, right, black]{\footnotesize $e_{15}$};
    \draw []
        (v0.center) -- (v1.center)
        (v1.center) -- (v2.center)
        (v2.center) -- (x.center)
        (y.center) -- (v14.center)
        (v13.center) -- (v15.center);
    \draw [|-|, double, blue]
        (x.center) -- (v3.center);
    \draw [|-|, double, teal]
        (v2.center) -- (v4.center);
    \draw [|-, double, violet]
        (v1.center) -- (v5.center);
    \draw [-, double, violet]
        (v1.center) -- (v8.center);
    \draw [-|, double, violet]
        (v5.center) -- (v6.center);
    \draw [-|, double, violet]
        (v5.center) -- (v7.center);
    \draw [-|, double, violet]
        (v8.center) -- (v9.center);
    \draw [-|, double, violet]
        (v8.center) -- (v10.center);
    \draw [|-, double, brown]
        (v0.center) -- (v11.center);
    \draw [-|, double, brown]
        (v11.center) -- (v12.center);
    \draw [-, double, brown]
        (v11.center) -- (v13.center);
    \draw [-|, double, brown]
        (v13.center) -- (y.center);
    \draw [black, fill=black]
        (v0) circle [radius=0.05]
        (v1) circle [radius=0.05]
        (v2) circle [radius=0.05]
        (v3) circle [radius=0.05]
        (v4) circle [radius=0.05]
        (v5) circle [radius=0.05]
        (v6) circle [radius=0.05]
        (v7) circle [radius=0.05]
        (v8) circle [radius=0.05]
        (v9) circle [radius=0.05]
        (v10) circle [radius=0.05]
        (v11) circle [radius=0.05]
        (v12) circle [radius=0.05]
        (v13) circle [radius=0.05]
        (v14) circle [radius=0.05]
        (v15) circle [radius=0.05];
    \end{tikzpicture}
    \caption{Tree $G$}
  \end{subfigure}
  \begin{subfigure}[b]{.45\textwidth}
    \centering 
    \begin{tikzpicture}[scale=0.9]
    \usetikzlibrary{positioning}
    \node at (0,0) (v0) {};
    \node at (0:0.5cm) (v1) {};
    \node at (0:1.0cm) (v2) {};
    \node[label=below:\footnotesize $z_3$] at (0:1.5cm) (v3) {};
    \node[label=below:\footnotesize $z_2$] at (0:2.0cm) (v4) {};
    \node at (0:2.5cm) (v5) {};
    \node at (0:3.0cm) (v6) {};
    \node at (0:3.5cm) (v7) {};
    \node at (0:4.0cm) (v8) {};
    \node at (0:4.5cm) (v9) {};
    \node[label=below:\footnotesize $z_1$] at (0:5.0cm) (v10) {};
    \node at (0:5.5cm) (v11) {};
    \node at (0:6.0cm) (v12) {};
    \node at (0:6.5cm) (v13) {};
    \node at (0:7.0cm) (v14) {};
    \node at (0:7.5cm) (v15) {};
    \node[label=below:\footnotesize $x$] at (0:1.25cm) (x) {};
    \node[label=below:\footnotesize $y$] at (0:6.75cm) (y) {};
    \draw [white]
        (v0.center) -- (v1.center)  node[pos=0.5, above, black]{\footnotesize $e_1$}
        (v1.center) -- (v2.center)  node[pos=0.5, above, black]{\footnotesize $e_2$}
        (v2.center) -- (v3.center)  node[pos=0.5, above, black]{\footnotesize $e_3$}
        (v3.center) -- (v4.center)  node[pos=0.5, above, black]{\footnotesize $e_4$}
        (v4.center) -- (v5.center)  node[pos=0.5, above, black]{\footnotesize $e_5$}
        (v5.center) -- (v6.center)  node[pos=0.5, above, black]{\footnotesize $e_6$}
        (v6.center) -- (v7.center)  node[pos=0.5, above, black]{\footnotesize $e_7$}
        (v7.center) -- (v8.center)  node[pos=0.5, above, black]{\footnotesize $e_8$}
        (v8.center) -- (v9.center)  node[pos=0.5, above, black]{\footnotesize $e_9$}
        (v9.center) -- (v10.center)  node[pos=0.5, above, black]{\footnotesize $e_{10}$}
        (v10.center) -- (v11.center)  node[pos=0.5, above, black]{\footnotesize $e_{11}$}
        (v11.center) -- (v12.center)  node[pos=0.5, above, black]{\footnotesize $e_{12}$}
        (v12.center) -- (v13.center)  node[pos=0.5, above, black]{\footnotesize $e_{13}$}
        (v13.center) -- (v14.center)  node[pos=0.5, above, black]{\footnotesize $e_{14}$}
        (v14.center) -- (v15.center)  node[pos=0.5, above, black]{\footnotesize $e_{15}$};
    \draw []
        (v0.center) -- (x.center)
        (y.center) -- (v15.center);
    \draw [|-|, double, blue]
        (x.center) -- (v3.center);
    \draw [|-|, double, teal]
        (v3.center) -- (v4.center);
    \draw [|-|, double, violet]
        (v4.center) -- (v10.center);
    \draw [|-|, double, brown]
        (v10.center) -- (y.center);
    \draw [black, fill=black]
        (v0) circle [radius=0.05]
        (v1) circle [radius=0.05]
        (v2) circle [radius=0.05]
        (v3) circle [radius=0.05]
        (v4) circle [radius=0.05]
        (v5) circle [radius=0.05]
        (v6) circle [radius=0.05]
        (v7) circle [radius=0.05]
        (v8) circle [radius=0.05]
        (v9) circle [radius=0.05]
        (v10) circle [radius=0.05]
        (v11) circle [radius=0.05]
        (v12) circle [radius=0.05]
        (v13) circle [radius=0.05]
        (v14) circle [radius=0.05]
        (v15) circle [radius=0.05];
    \end{tikzpicture}
    \caption{Path graph $\phi(G)$}
  \end{subfigure}
\caption{(a) A tree $G$ of height $3$ and (b) its corresponding path graph $\phi(G)$ based on a depth-first search of $G$. A segment $[x, y]$ along the path graph (double lines) corresponds to at most \emph{four} connected pieces in the tree.} \label{fig:psn_thm52}
\end{figure}

\begin{thm} \label{thm:diameter}
Let $G$ be a connected graph, and let $d$ be the diameter of a spanning tree of $G$ with the minimum diameter. Then there exists a bijection $\phi$ that can be computed in polynomial time such that $\psn{G, \phi} \leq \left\lceil d/2 \right\rceil + 2$.
\end{thm}

\begin{proof}
A spanning tree $T$ of $G$ with the minimum diameter, and a vertex $r$ that minimizes its maximum distance to any vertex in $T$, can be found in polynomial time \citep{HassinTa95}.\footnote{$r$ is known as an \emph{absolute $1$-center of $G$} \citep{HassinTa95}.} Let $T$ be rooted at $r$.
We claim that the height of $T$ is at most $\left\lceil d/2 \right\rceil$.
To see this, let $v_1,\dots,v_k$ be the children of $r$, and assume for contradiction that some node $v$ is of distance greater than $\left\lceil d/2 \right\rceil$ from $r$; without loss of generality, suppose $v$ is a descendant of $v_1$.
If there exists a descendant $w$ of some node in $\{v_2,\dots,v_k\}$ of distance at least $\left\lceil d/2 \right\rceil$ from $r$, then the path from $v$ to $w$ must pass through $r$ since $T$ is a tree, and the distance between $v$ and $w$ is greater than $\left\lceil d/2 \right\rceil + \left\lceil d/2 \right\rceil \ge d$, contradicting the fact that $d$ is the diameter of $T$.
Therefore, every descendant of $v_2,\dots,v_k$ must be of distance less than $\left\lceil d/2 \right\rceil$ from $r$.
It follows that $r$ is not a node that minimizes its maximum distance to any vertex in $T$, because $v_1$ reduces this quantity by one.
This contradicts the definition of $r$.
Hence, the height of $T$ is at most $\left\lceil d/2 \right\rceil$.

Now, construct the tree $T'$ as follows: start with all vertices and edges in $T$ and with root vertex $r$; then for each edge $e = [u, v]$ in $G \setminus T$, add a new leaf vertex $v'$ as a child of $u$ in $T'$, where the edge $[u, v']$ in $T'$ corresponds to the edge $e$ in $G$. 
If $v$ appears multiple times across different edges in $G \setminus T$, then a \emph{new} leaf vertex is created each time.
The constructed graph $T'$ is a tree with corresponding edges in $G$, and has height at most one more than the height of $T$; furthermore, any connected piece in $T'$ corresponds to a connected piece in $G$. 
Let $\phi'$ be the corresponding bijection of the edges from $G$ to $T'$.
By \Cref{thm:tree} applied to $T'$ with height $h \le \left\lceil d/2 \right\rceil + 1$, there exists a bijection $\phi''$ that can be computed in polynomial time such that $\psn{T', \phi''} \leq \left\lceil d/2 \right\rceil + 2$.
Then, we have $\psn{G, \phi'' \circ \phi'} \leq \psn{T', \phi''} \leq \left\lceil d/2 \right\rceil + 2$, as desired.
\end{proof}

\section{Conclusion and Future Work} \label{sec:conclusion}

In this paper, we have studied the existence and computation of approximately envy-free allocations in graphical cake cutting.
For general graphs, we devised polynomial-time algorithms for computing a $1/2$-additive-envy-free allocation in the case of non-identical valuations and a $(2+\epsilon)$-envy-free allocation in the case of identical valuations.
For star graphs, our efficient algorithms provide a multiplicative envy factor of $3+\epsilon$ for non-identical valuations and $2$ for identical valuations.
Our bounds in the case of identical valuations are (essentially) tight.
We also explored envy-freeness guarantees when the connectivity assumption is relaxed, through the notion of path similarity number.

An interesting question left open by our work is whether a connected allocation with a constant multiplicative approximation of envy-freeness can be guaranteed for general graphs and non-identical valuations---the techniques that we developed for star graphs (\Cref{subsec:nonidenvalue-star}) do not seem sufficient for answering this question.
If a non-connected allocation is allowed, then tightening the bounds for the path similarity number (\Cref{sec:beyond}) will reduce the number of connected pieces each agent receives to achieve the same envy-freeness approximation.
Another intriguing direction for non-connected allocations is whether we can obtain improved envy bounds using a different approach than converting the graph into a path.

\section*{Acknowledgments}

This work was partially supported by the Singapore Ministry of Education under grant number MOE-T2EP20221-0001 and by an NUS Start-up Grant.
We thank the anonymous IJCAI 2023 and Discrete Applied Mathematics reviewers for their valuable feedback.

\bibliographystyle{plainnat}
\bibliography{main}

\appendix

\section{Proof of \texorpdfstring{\Cref{prop:relationships}}{Proposition 1}} \label{app:relationships}
\begin{unnumberedthm}{\Cref{prop:relationships}}
Let $\mathcal{A}$ be an allocation for $n \geq 2$ agents, and let $\alpha \geq 1$. 
\begin{itemize}
    \item If $\mathcal{A}$ is $\alpha$-EF, then it is $\left(\alpha - \frac{\alpha - 1}{n}\right)$-proportional.
    \item If $\mathcal{A}$ is $\alpha$-EF, then it is $\left( \frac{\alpha - 1}{\alpha + 1} \right)$-additive-EF.
    \item If $\mathcal{A}$ is $\alpha$-proportional, then it is $\left(1 - \frac{2}{\alpha n}\right)$-additive-EF.
\end{itemize}
\end{unnumberedthm}

\begin{proof}
Let $i \in N$.
We prove the three statements in turn.
\begin{itemize}
    \item Suppose that $\mathcal{A}$ is $\alpha$-EF, and let $\xi = \val[i]{A_i}$.
    The share of every other agent $j \in N$ is worth at most $\alpha \xi$ to agent~$i$, so 
    \[
    1 = \sum_{j \in N} \val[i]{A_j} \leq \xi + (n-1) \alpha \xi = \xi (\alpha n - \alpha + 1).
    \]
    This gives 
    \[
    \xi \geq \frac{1}{\alpha n - \alpha + 1} = \frac{1}{\left(\alpha - \frac{\alpha - 1}{n}\right)n},
    \]
    which establishes $(\alpha - \frac{\alpha - 1}{n})$-proportionality.
    \item Suppose that $\mathcal{A}$ is $\alpha$-EF, and let $j \in N$. 
    If $\val[i]{A_j} > \frac{\alpha}{\alpha + 1}$, then $\val[i]{A_i} < 1 - \frac{\alpha}{\alpha + 1} = \frac{1}{\alpha} \left( \frac{\alpha}{\alpha + 1} \right) <\frac{\val[i]{A_j}}{\alpha}$, and the allocation is not $\alpha$-EF, a contradiction.
    Hence, $\val[i]{A_j} \leq \frac{\alpha}{\alpha + 1}$. 
    By definition of $\alpha$-EF we have $\val[i]{A_i} \geq \val[i]{A_j} / \alpha$, and so 
    \begin{align*}
        \val[i]{A_j} - \val[i]{A_i} &\leq \val[i]{A_j} - \frac{\val[i]{A_j}}{\alpha} \\
        &= \left(1 - \frac{1}{\alpha} \right) \val[i]{A_j} 
        \leq \left(1 - \frac{1}{\alpha} \right) \left( \frac{\alpha}{\alpha + 1} \right) 
        = \frac{\alpha - 1}{\alpha + 1};
    \end{align*}
    this shows $\left( \frac{\alpha - 1}{\alpha + 1} \right)$-additive-EF.
    \item Suppose that $\mathcal{A}$ is $\alpha$-proportional, and let $j \in N$. 
    Then $\val[i]{A_i} \geq \frac{1}{\alpha n}$ and $\val[i]{A_j} \leq 1 - \val[i]{A_i} \leq 1 - \frac{1}{\alpha n}$, so 
    \[
    \val[i]{A_j} - \val[i]{A_i} \leq 1 - \frac{2}{\alpha n},
    \]
    proving $(1 - \frac{2}{\alpha n})$-additive-EF.
\end{itemize}
This completes the proof.
\end{proof}

\section{\texorpdfstring{\textsc{Divide}}{Divide}} \label{app:divide}

In this appendix, we describe the algorithm \textsc{Divide} (\Cref{alg:divide}), which takes as input a connected subgraph~$H$ worth $\beta_0$ to some agent in $N' \subseteq N$, a positive threshold $\beta \leq \beta_0$, and a root vertex $r$ of $H$. 
The output of \textsc{Divide} satisfies the conditions stated in \Cref{lem:divide}.

\begin{algorithm}[!h]
    \caption{\textsc{Divide}$(H, N', \beta, r)$.}
    \label{alg:divide}
    \textbf{Input}: Connected subgraph $H$, set of agents $N'$, threshold $\beta \in (0, \beta_0]$ where $\beta_0 = \val[i]{H}$ for some $i \in N'$, vertex $r$ of~$H$. \\
    \textbf{Output}: Graphs $H_1$ and $H_2$. \\
    \vspace{-3.5mm}
    \begin{algorithmic}[1]
        \STATE $L \leftarrow \emptyset$ (a list of vertices) \label{ln:divide-decycle1}
        \WHILE{there exists a cycle in $H$}
            \STATE $[v_1, v_2] \leftarrow$ any edge on the cycle \label{ln:cycle1}
            \STATE add a new vertex $v_2'$ to the graph
            \STATE append $v_2'$ to $L$
            \STATE $[v_1, v_2'] \leftarrow [v_1, v_2]$
            \STATE delete edge $[v_1, v_2]$ from the graph \label{ln:cycle2}
        \ENDWHILE \ {\scriptsize (this converts $H$ into a tree)} \label{ln:divide-decycle2}
        \STATE $v \leftarrow r$ \label{ln:divide-split1}
        \WHILE{there exists $i \in N'$ and a child vertex $w$ of $v$ such that $\val[i]{S_w} \geq \beta$}
            \STATE $v \leftarrow$ any child vertex $w$ of $v$ with $\val[i]{S_w} \geq \beta$
        \ENDWHILE \ {\scriptsize ($S_v$ is worth at least $\beta$ to some agent)}
        \IF{there exists $i \in N'$ and a child vertex $w$ of $v$ such that $\val[i]{S_{v,w}} \geq \beta$}
            \STATE $w^* \leftarrow$ any child vertex $w$ of $v$ with $\val[i]{S_{v,w}} \geq \beta$ \label{ln:divide-if1}
            \STATE $x \leftarrow$ the point in $[w^*, v]$ closest to $w^*$ such that there exists $i \in N'$ with $\val[i]{S_{x,w^*}} = \beta$ 
            \STATE $H_1 \leftarrow S_{x,w^*}$
            \STATE $H_2 \leftarrow H \setminus H_1$ (add the point $x$ to $H_2$) \label{ln:divide-if2}
        \ELSE
            \STATE $C \leftarrow \emptyset$ (a list of child vertices of $v$) \label{ln:divide-else1}
            \WHILE{$\sum_{w \in C} \val[i]{S_{v,w}} < \beta$ for all $i \in N'$}
                \STATE add to $C$ any child vertex of $v$ that is not yet in $C$
            \ENDWHILE
            \STATE $H_1 \leftarrow \bigcup_{w \in C} S_{v, w}$
            \STATE $H_2 \leftarrow H \setminus H_1$ (add the point $v$ to $H_2$) \label{ln:divide-else2}
        \ENDIF \label{ln:divide-split2}
        \WHILE{$L \neq \emptyset$} \label{ln:divide-recycle1}
            \STATE reverse lines \ref{ln:cycle1} to \ref{ln:cycle2} based on the last element of $L$
            \STATE remove from $L$ its last element
            \STATE adjust $H_1$ and $H_2$ accordingly
        \ENDWHILE \ {\scriptsize (this converts the graph back to the original $H$)} \label{ln:divide-recycle2}
        \STATE \textbf{return} $(H_1, H_2)$
    \end{algorithmic}
\end{algorithm}

We will work with rooted trees. 
Let $T = (V, E)$ be a rooted tree. 
For a vertex $v \in V$, let $S_v$ be the subtree at~$v$, that is, $S_v$ is the subgraph induced by $v$ and all of its descendants in $T$.
For a child vertex $w$ of $v$, let $S_{v, w}$ be the subgraph induced by $v$, $w$, and all descendants of $w$ in $T$.

We begin by converting the graph $H$ into a tree (Lines \ref{ln:divide-decycle1} to \ref{ln:divide-decycle2}). 
As long as $H$ contains a cycle, select an edge $[v_1, v_2]$ belonging to the cycle, add a new vertex $v_2'$, and replace the edge $[v_1, v_2]$ with the edge $[v_1, v_2']$ while keeping the remaining edges of the graph. 
Note that the new edge created should have the same value to each agent as the one it replaces. 
This procedure decreases the number of cycles in the graph by at least one. 
Therefore, by repeating this procedure, $H$ eventually becomes a tree.
Note that any connected share of this tree corresponds to a connected share in the original graph with the same value for every agent.

We now split the tree $H$ into two subtrees $H_1$ and $H_2$ (Lines \ref{ln:divide-split1} to \ref{ln:divide-split2}). 
To do so, traverse the tree from the root vertex~$r$ until some vertex~$v$ is reached such that the subtree~$S_v$ is worth at least $\beta$ to some agent in $N'$, while the subtree~$S_w$ at each child vertex~$w$ of $v$ is worth less than~$\beta$ to all agents in $N'$. 
Since $S_r$ is worth at least $\beta$ to some agent while $S_z$ is worth $0$ to all agents for every leaf vertex $z$, there must be some vertex $v$ where the condition holds. 
We consider two cases.

\begin{itemize}
    \item \textbf{Case 1 (Lines \ref{ln:divide-if1} to \ref{ln:divide-if2}): There exists an agent $i \in N'$ and a child vertex~$w$ of $v$ such that $\val[i]{S_{v, w}} \geq \beta$.} \\
    Let $w^*$ be one such child vertex. By assumption, the subtree $S_{w^*}$ is worth less than~$\beta$ to all agents. 
    Find the point $x \in [w^*, v]$ closest to $w^*$ such that $S_{x, w}$ is worth exactly $\beta$ to some agent $i^*$ and at most $\beta$ to all other agents.
    (Here we abuse notation slightly and treat $x$ as a vertex.)
    Let $S_{x, w}$ be the first share, and the remaining portion of $H$ be the second share. Note that a new vertex $x \in [w^*, v]$ is created and belongs to both shares.
    The first share is worth $\beta$ to some agent and at most $\beta$ to every agent.
    \item \textbf{Case 2 (Lines \ref{ln:divide-else1} to \ref{ln:divide-else2}): 
    For all agents $i \in N'$ and all child vertices $w$ of~$v$, we have $\val[i]{S_{v, w}} < \beta$.} \\
    Initialize $H_1$ to be a graph with only vertex $v$. 
    For each child vertex $w$ of $v$, iteratively add $S_{v, w}$ to $H_1$ until $H_1$ is worth at least $\beta$ to some agent $i^*$. 
    Let $H_1$ be the first share, and the remaining portion of $H$ be the second share. 
    Note that the vertex $v$ belongs to both shares.
    The first share is worth at least $\beta$ to some agent and less than $2 \beta$ to all agents, because $H_1$ is worth less than~$\beta$ \emph{before} the last $S_{v, w}$ is added, and the last $S_{v, w}$ is also worth less than $\beta$, so their combined value is less than $2 \beta$ to every agent. 
\end{itemize}

Lastly, we convert the tree back to the original graph (Lines \ref{ln:divide-recycle1} to \ref{ln:divide-recycle2}). This is done simply by reversing the steps taken to convert the original graph into a tree. Along the way, we also adjust the two shares accordingly. At each step, each of the two shares is connected since the vertices that are removed are always leaf vertices.

\end{document}